\documentclass[a4paper,UKenglish]{lipics}

\usepackage{amsmath,accents,amssymb,amsfonts}
\usepackage{floatflt,thm-restate}

\usepackage{xcolor,tikz}

\theoremstyle{plain}

\def\net{\mathcal{N}}
\def\calN{\mathcal{N}}
\def\calG{\mathcal{G}}
\def\calC{\mathcal{C}}
\def\calD{\mathcal{D}}
\def\calR{\mathcal{R}}
\def\calL{\mathcal{L}}
\def\calI{\mathcal{I}}
\def\calF{\mathcal{F}}
\def\calU{\mathcal{U}}
\def\calS{\mathcal{S}}
\def\bbN{\mathbb{N}}
\def\bbR{\mathbb{R}}
\def\bbQ{\mathbb{Q}}
\def\bbL{\mathbb{L}}
\def\bbZ{\mathbb{Z}}
\def\xy{\overline{\mathfrak{m}}}
\def\bfxy{\mathbf{xy}}
\def\bfC{\textbf{C}}
\def\bfA{\textbf{A}}
\def\bfK{\textbf{K}}
\def\bfB{\textbf{B}}
\def\bfS{\textbf{S}}
\def\bfE{\textbf{E}}
\def\bfW{\textbf{W}}
\def\bfw{\textbf{w}}
\def\bfv{\textbf{v}}
\def\bfG{\textbf{G}}
\def\bfP{\textbf{P}}
\def\bfU{\textbf{U}}
\def\bfZ{\textbf{Z}}
\def\bfD{\textbf{D}}
\def\nat{\mathbb{N}}
\def\pideux{$\mathsf{\Pi}^2$}
\def\pitrois{$\mathsf{\Pi}^3$}
\def\wit{\mathit{wit}}
\def\ext{\mathsf{ext}}
\def\pot{\mathsf{pot}}
\def\cin{\mathsf{cin}}
\def\Live{\mathsf{Live}}
\def\POT{\mathsf{POT}}
\def\Inv{\mathsf{Inv}}
\newcommand{\hatt}[1]{\hat{\vphantom{\rule{1pt}{5.5pt}}\smash{\hat{#1}}}}

\title{Unbounded product-form Petri nets\footnote{This
    work has been supported by ERC project EQualIS
    (FP7-308087).}}
    
\author[1]{Patricia Bouyer}
\author[2]{Serge Haddad}
\author[1]{Vincent Jug\'e}

\affil[1]{LSV, CNRS, ENS Paris-Saclay, Universit\'e Paris-Saclay,
  France} 

\affil[2]{LSV, CNRS, ENS Paris-Saclay, Inria, Universit\'e
  Paris-Saclay, France}

\Copyright{Patricia Bouyer, Serge Haddad, Vincent Jug\'e}
\subjclass{D.2.2 -- Petri nets}
\DOIPrefix{}
\keywords{Performance evaluation, infinite-state systems, Petri nets,
 steady-state distribution}

\EventEditors{}
\EventNoEds{0}
\EventLongTitle{}
\EventShortTitle{}
\EventAcronym{}
\EventYear{2017}
\EventDate{}
\EventLocation{}
\EventLogo{}
\SeriesVolume{}
\ArticleNo{1}

\begin{document}

\maketitle

\begin{abstract}
  Computing steady-state distributions in infinite-state stochastic
  systems is in general a very difficult task. Product-form Petri nets
  are those Petri nets for which the steady-state distribution can be
  described as a natural product corresponding, up to a normalising
  constant, to an exponentiation of the markings. However, even though
  some classes of nets are known to have a product-form distribution,
  computing the normalising constant can be hard. The class of
  (closed) $\mathsf{\Pi}^3$-nets has been proposed in an earlier work, for
  which it is shown that one can compute the steady-state distribution
  efficiently. However these nets are bounded. In this paper, we
  generalise queuing Markovian networks and closed $\mathsf{\Pi}^3$-nets to
  obtain the class of open $\mathsf{\Pi}^3$-nets, that generate
  infinite-state systems. We show interesting properties of these nets:
  (1) we prove that liveness can be decided in polynomial time, and
  that reachability in live $\mathsf{\Pi}^3$-nets can be decided in
  polynomial time; (2) we show that we can decide ergodicity of such
  nets in polynomial time as well; (3) we provide a pseudo-polynomial
  time algorithm to compute the normalising constant.
\end{abstract}

\section{Introduction}

\subparagraph*{Quantitative analysis of stochastic infinite-state systems.}
Performance
measures of stochastic systems can be roughly classified in two
categories: those related to the transient behaviour as expressed by a
temporal logic formula and those related to the long-run behaviour
whose main measure is the steady-state distribution (when it exists).

There are different relevant questions concerning
the steady-state distribution of infinite-state systems: (1) given a
state and a threshold, one can ask whether the steady-state
probability of this state is above (or below) the threshold; (2) given
a state, one can compute the steady-state probability of this state,
either in an exact way, or in an accurate approximate way; and (3) one
can give a symbolic representation of both the set of reachable states
and its associated distribution (or an accurate approximation
thereof).

Clearly the last question is the most difficult one, and the first
breakthrough in that direction has been obtained in the framework of
open queuing Markovian networks: in those systems,
the measure of a state is obtained as the product of terms, where each
term is related to the parameters of a queue (service and visit rate)
and the number of clients in the queue~\cite{jack63}.  In order to get
a probability distribution over the set of reachable states, this
product is normalised by a constant, whose computation is easy when
the service rates of the queues do not depend on the number of
clients. This work has been adapted to closed networks, and
the main contribution in~\cite{Gordon67} consists in computing the
normalising constant without enumerating the (finite) reachability
set, leading to an algorithm which runs in polynomial-time w.r.t. the
size of the network and the number of clients.  Later, Markov chains
generated by a stochastic Petri net with a single unbounded place
(that is, quasi-birth death processes) have been
investigated~\cite{FLORIN1986}, and an algorithm which approximates up
to arbitrary precision
the steady-state distribution has been proposed; however the
complexity of the algorithm is very high, since it requires the
computation of the finite reachability sets of some subnets, whose
size may be non primitive recursive. More recently, the abstract
framework of (infinite-state) Markov chains with a finite eager
attractor (e.g. probabilistic lossy channel systems) has been used to
develop an algorithm which approximates up to arbitrary precision the
steady-state distribution as well~\cite{AbdullaHMS06}, but there is no
complexity bound for the algorithm.

\subparagraph*{Product-form Petri nets.}  While queuing
networks are very interesting since they allow for an explicit
representation of the steady-state distribution, they
lack two important features available in Petri nets~\cite{gspnbook,molloy},
which are very relevant for modelling concurrent systems:
resource competition and process synchronisation.  So very soon
researchers have tried to get the best of the two formalisms and they
have defined subclasses of Petri nets for which one can establish
product-form steady-state distributions. 
Historically, solutions have been based on
purely behavioural properties (i.e. by an analysis of the reachability
graph) like in~\cite{lazar-robert91}, and then progressively have
moved to more and more structural
characterisations~\cite{li_geor92,CHT}.  Building on the work
of~\cite{CHT}, \cite{HMSS-pe59(4)} has established the first
purely structural condition for which a product-form steady-state
distribution exists, and designed a polynomial-time algorithm to check
for the condition (see also \cite{MairesseNguyen10} for an alternative
characterisation). These nets are called $\mathsf{\Pi}^2$-nets.  However the
computation of the normalising constant remains a difficult open issue
since a naive approach in the case of finite-state $\mathsf{\Pi}^2$-nets
would require to enumerate the potentially huge reachability
state-space.  Furthermore, the lower bounds shown in~\cite{HMN-fi13}
for behavioural properties of $\mathsf{\Pi}^2$-nets
strongly suggest that the computation of the normalising constant
can probably not be done in an efficient way.
In~\cite{CHT,SB}, the authors introduce semantical classes
of product-form Petri nets for which this constant is computed
in pseudo-polynomial time. However their approach suffers
two important drawbacks: (1) checking whether a net fufills
this condition as least as hard as the reachability problem,
and (2) the only syntactical class for which this condition
is fulfilled boils down to queuing networks.

To overcome this problem, the model of $\mathsf{\Pi}^3$-nets is defined
in~\cite{HMN-fi13} as
a subclass of $\mathsf{\Pi}^2$-nets obtained by structuring the
synchronisation between concurrent activity flows in layers. This
model strictly generalises closed product-form queuing networks (in
which there is a single activity flow).  Two interesting properties of
those nets is that liveness for $\mathsf{\Pi}^3$-nets and reachability for
live $\mathsf{\Pi}^3$-nets can both be checked in polynomial time.
Furthermore, from a quantitative point-of-view, the normalising
constant of the 
steady-state distribution can be
efficiently computed using an elaborated dynamic programming
algorithm.

Product-form Petri nets have been applied for the specification and
analysis of complex systems. From a modelling point-of-view,
compositional approaches have been proposed~\cite{Balsamobis,BM2009}
as well as hierarchical ones~\cite{Harrison2011}. Application fields
have also been identified, for instance, hardware design and more
particularly RAID storage~\cite{Harrison2011}, or software
architectures~\cite{BalsamoMarin2011}.

\subparagraph*{Our contributions.}
Unfortunately
$\mathsf{\Pi}^3$-nets generate finite-state systems.  Here we address this
problem by introducing and studying open $\mathsf{\Pi}^3$-nets. Informally,
an open $\mathsf{\Pi}^3$-net has a main activity flow, which roughly
corresponds to an open queuing network, and has other activity flows
which are structured as in (standard, or closed) $\mathsf{\Pi}^3$-nets.
More precisely, in the case of a single activity flow this model is
exactly equivalent to an open queuing network, but the general model is
enriched with other activity flows, raising difficult computation
issues. In particular, several places may in general be unbounded.
Open $\mathsf{\Pi}^3$-nets are particularly appropriate when modelling,
in an open environment, protocols and softwares designed in layers.
In adddition, they allow to specify dynamical management of resources
where processes may produce and consume them with no \emph{a priori} upper bound
on their number.  Our results on open $\mathsf{\Pi}^3$-nets can be summarised as
follows:
\begin{itemize}
\item We first establish that the liveness problem 
  can be solved in polynomial time, and that the boundedness as well
  as the reachability problem in live nets can also be solved in
  polynomial time. On the other side, we show that the unboundedness,
  the reachability and even the covering problem become NP-hard
  without the liveness assumption.
\item Contrary to the case of closed $\mathsf{\Pi}^3$-nets, open
  $\mathsf{\Pi}^3$-nets may not be ergodic (that is, there may not exist a
  steady-state distribution). We design a polynomial-time algorithm to
  decide ergodicity of an open $\mathsf{\Pi}^3$-net.
\item Our main contribution is the computation of the normalising
  constant for ergodic live $\mathsf{\Pi}^3$-nets. Our procedure combines
  symbolic computations and dynamic programming. The time complexity
  of the algorithm is polynomial w.r.t. the size of the structure of
  the net and the maximal value of integers occuring in the
  description of the net (thus pseudo-polynomial).
  As a side result, we improve the complexity for computing the
  normalising constant of closed $\mathsf{\Pi}^3$-nets that was given
  in~\cite{HMN-fi13} (the complexity was the same, but was assuming
  that the number of activity flows is a constant).
\end{itemize}

In Section~\ref{sec:formalism}, we introduce and illustrate
product-form nets, and recall previous results.  In
Section~\ref{sec:qualitative}, we focus on qualitative behavioural
properties, while quantitative analysis is developed in
Section~\ref{sec:quantitative}.
All proofs are postponed to the appendix.

\section{Product-form Petri nets}
\label{sec:formalism}

\subparagraph*{Notations.} Let $A$ be a matrix over $I\times J$, one
denotes $A(i,j)$ the item whose row index is $i$ and column index is
$j$. When $I$ and $J$ are disjoint, $W(k)$ denotes the row
(resp. column) vector indexed by $k$ when $k \in I$ (resp. $k \in J$).
Given a real vector $\mathbf{v}$ indexed by $I$ its norm, denoted
$\|\mathbf{v}\|$, is defined by $\|\mathbf{v}\|=\sum_{i\in I}
|\mathbf{v}(i)|$.  Sometimes, one writes $\mathbf{v}_i$ for
$\mathbf{v}(i)$.  Given two vectors $\mathbf{v}$, $\mathbf{w}$ indexed
by $I$ their scalar product denoted $\mathbf{v} \cdot \mathbf{w}$ is
defined by $\sum_{i\in I} \mathbf{v}_i \mathbf{w}_i$.  Finally, if
$\mathbf{v}$ is a vector over $I$, we define its support as
$\mathsf{Supp}(\mathbf{v}) = \{i \in I \mid \mathbf{v}(i) \ne 0\}$.

We briefly recall Petri nets and stochastic Petri nets.  The state of
a Petri net, called a \emph{marking} is defined by the number of
\emph{tokens} contained in every \emph{place}.  A Petri net models
concurrent activities by \emph{transitions} whose enabling requires
tokens to be consumed in some places and then tokens to be produced in
some places.

\begin{definition}[Petri net]
  \label{def:petrinet}
  A {\em Petri net} is a tuple $\mathcal{N}= (P,T,W^-,W^+)$ where:
\begin{itemize}
  \item $P$ is a finite set of {\em places};
  \item $T$ is a finite set of {\em transitions}, disjoint from $P$;
  \item $W^-$ and $W^+$ are $P \times T$ matrices with coefficients
    in $\mathbb{N}$.
\end{itemize}
\end{definition}

$W^-$ (resp. $W^+$) is called the backward (resp. forward) incidence
matrix, $W^-(p,t)$ (resp. $W^+(p,t)$) specifies the number of tokens
consumed (resp. produced) in place $p$ by the \emph{firing} of
transition $t$, and $W^-(t)$ (resp. $W^+(t)$) is the $t$-th column of
$W^-$ (resp. $W^+$).  One assumes that for all $t\in T$, $W^-(t)\neq
W^+(t)$ (i.e. no useless transition) and for all $t'\neq t$, either
$W^-(t)\neq W^-(t')$ or $W^+(t)\neq W^+(t')$ (i.e. no duplicated
transition); this will not affect our results.

A \emph{marking} of $\mathcal{N}$ is a vector of $\mathbb{N}^P$; in the sequel
  we will often see $m$ as a multiset ($m(p)$ is then the number of
  occurrences of $p$), or as a symbolic sum $\sum_{p \in P \mid
    m(p)>0} m(p)\, p$. The symbolic sum $\sum_{p \in P'} p$
will be more concisely written $P'$.
    Transition $t$ is \emph{enabled}
by marking $m \in \mathbb{N}^P$ if for all $p\in P$, $m(p)\geqslant
W^-(p,t)$. When enabled, its firing leads to the marking $m'$ defined
by: for all $p \in P$, $m'(p)=m(p)-W^-(p,t)+W^+(p,t)$. This firing is
denoted by $m\xrightarrow{t} m'$.  The \emph{incidence matrix}
$W=W^+-W^-$ allows one to rewrite the marking evolution as $m'=m+W(t)$
if $m\geqslant W^-(t)$.
Given an initial marking $m_0 \in \mathbb{N}^P$, the \emph{reachability set}
$\mathcal{R}_\mathcal{N}(m_0)$ is the smallest set containing $m_0$ and closed
under the firing relation.  When no confusion is possible one denotes
it more concisely by $\mathcal{R}(m_0)$. Later, if $m \in \mathcal{R}(m_0)$, we
may also write $m_0 \to^* m$. We will call $(\mathcal{N},m_0)$ a
\emph{marked} Petri net

\begin{figure}[htbp]
\begin{center}
\begin{tikzpicture}[xscale=0.7,yscale=0.7]
\path (3,10) node[draw,circle,inner sep=2pt,minimum size=3mm,color=red] (p0) 
[label=below:$p_2$] {};
\path (5,8) node[draw,circle,inner sep=2pt,minimum size=3mm,color=red] (p1) 
[label=left:$p_1$] {};
\path (1,6) node[draw,circle,inner sep=2pt,minimum size=3mm,fill=gray!50,color=red] (pext) 
[label=left:$p_{\mathsf{ext}}$] {};
\path (3,4) node[draw,circle,inner sep=2pt,minimum size=3mm,color=red] (p2) 
[label=above:$p_0$] {};

\path (5,9) node[draw,rectangle,inner sep=2pt,minimum size=3mm,color=red] (t0) 
[label=left:$t_0$] {};
\path (1,7) node[draw,rectangle,inner sep=2pt,minimum size=3mm,color=red] (t1) 
[label=left:$t_1$] {};
\path (5,7) node[draw,rectangle,inner sep=2pt,minimum size=3mm,color=red] (t2) 
[label=left:$t_2$] {};
\path (1,5) node[draw,rectangle,inner sep=2pt,minimum size=3mm,color=red] (t3) 
[label=left:$t_3$] {};
\path (3,3) node[draw,rectangle,inner sep=2pt,minimum size=3mm,color=red] (t4) 
[label=right:$t_4$] {};

\draw[->,>=stealth,color=red] (p0) --(5,9.5)--(t0);
\draw[->,>=stealth,color=red] (p0) --(1,9.5)-- (t1);
\draw[->,>=stealth,color=red] (t0) -- (p1);
\draw[->,>=stealth,color=red] (p1) -- (t2);
\draw[->,>=stealth,color=red] (t1) -- (pext);
\draw[->,>=stealth,color=red] (pext) -- (t3);

\draw[->,>=stealth,color=red] (t2) --(5,4.5)-- (p2);
\draw[->,>=stealth,color=red] (t3) --(1,4.5)-- (p2);

\draw[->,>=stealth,color=red] (p2) -- (t4);
\draw[->,>=stealth,color=red] (t4) -- (-1,3)--(-1,10) --(p0);

\path (7,10) node[draw,circle,inner sep=2pt,minimum size=3mm,color=green!60!black] (q0) 
[label=below right:$q_3$] {$\bullet$};
\path (7,8) node[draw,circle,inner sep=2pt,minimum size=3mm,color=green!60!black] (q1) 
[label=right:$q_2$] {};
\path (7,6) node[draw,circle,inner sep=2pt,minimum size=3mm,color=green!60!black] (q2) 
[label=right:$q_1$] {};
\path (7,4) node[draw,circle,inner sep=2pt,minimum size=3mm,color=green!60!black] (q3) 
[label=right:$q_0$] {};

\path (7,9) node[draw,rectangle,inner sep=2pt,minimum size=3mm,color=green!60!black] (t5) 
[label=left:$t_5$] {};
\path (7,7) node[draw,rectangle,inner sep=2pt,minimum size=3mm,color=green!60!black] (t6) 
[label=left:$t_6$] {};
\path (7,5) node[draw,rectangle,inner sep=2pt,minimum size=3mm,color=green!60!black] (t7) 
[label=left:$t_7$] {};
\path (7,3) node[draw,rectangle,inner sep=2pt,minimum size=3mm,color=green!60!black] (t8) 
[label=left:$t_8$] {};

\draw[->,>=stealth,color=green!60!black] (q0) -- (t5);
\draw[->,>=stealth,color=green!60!black] (t5) -- (q1);
\draw[->,>=stealth,color=green!60!black] (q1) -- (t6);
\draw[->,>=stealth,color=green!60!black] (t6) -- (q2);
\draw[->,>=stealth,color=green!60!black] (q2) -- (t7);
\draw[->,>=stealth,color=green!60!black] (t7) -- (q3);
\draw[->,>=stealth,color=green!60!black] (q3) -- (t8);
\draw[->,>=stealth,color=green!60!black] (t8) -- (13,3)-- (13,10)--(q0);

\path (9,7) node[draw,circle,inner sep=2pt,minimum size=3mm,color=blue] (r0) 
[label=right:$r_0$] {$\bullet$};
\path (11,7) node[draw,circle,inner sep=2pt,minimum size=3mm,color=blue] (r1) 
[label=left:$r_1$] {};

\path (10,8) node[draw,rectangle,inner sep=2pt,minimum size=3mm,color=blue] (t9) 
[label=above:$t_9$] {};
\path (10,6) node[draw,rectangle,inner sep=2pt,minimum size=3mm,color=blue] (t10) 
[label=below:$t_{10}$] {};

\draw[->,>=stealth,color=blue] (r0) -- (t10);
\draw[->,>=stealth,color=blue] (t10) -- (r1);
\draw[->,>=stealth,color=blue] (r1) -- (t9);
\draw[->,>=stealth,color=blue] (t9) -- (r0);

\draw[->,>=stealth,<->,gray] (t5) -- (r0);
\draw[->,>=stealth,<->,gray] (t6) -- (r0);
\draw[->,>=stealth,gray] (r0) -- (t7);
\draw[->,>=stealth,gray] (t8) --(9,5) --(r0);

\draw[->,>=stealth,gray] (t0) -- (q1);
\draw[->,>=stealth,gray] (q1) -- (t2);

\draw[->,>=stealth,gray] (t1) -- (q2);
\draw[->,>=stealth,gray] (q2) -- (t3);

\draw[->,>=stealth,gray] (t4) -- (3,2.6)node[pos=1,right] {$3$}-- (-1.5,2.6)-- 
(-1.5,10.5)--(7,10.5)--(q0);

\draw[->,>=stealth,gray] (q0) -- (t0) node[pos=.5,below]{$3$};
\draw[->,>=stealth,gray] (q0) --(5,10) -- (t1) node[pos=.5,below]{$3$};

\end{tikzpicture}
\caption{A marked Petri net (initial marking: $q_3+r_0$).}
\label{fig:PN}
\end{center}
\end{figure}
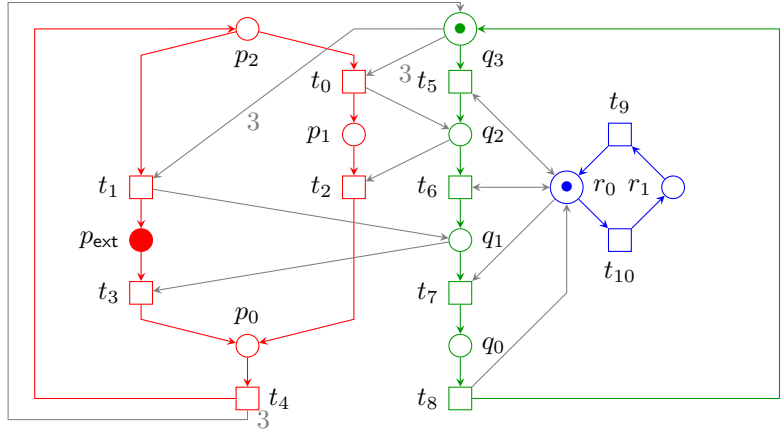

From a qualitative point of view, one is interested by several
standard relevant properties including reachability. \emph{Liveness} means that the modelled
system never loses its capacity: for all $t\in T$ and $m\in
\mathcal{R}(m_0)$, there exists $m' \in \mathcal{R}(m)$ such that $t$ is enabled
in $m'$.  \emph{Boundedness} means that the modelled system is a
finite-state system: there exists $B \in \mathbb{N}$ such that for all $m\in
\mathcal{R}(m_0)$, $\|m\|\leqslant B$. While decidable, these properties are
costly to check: (1) Reachability is
EXPSPACE-hard in general and PSPACE-complete for 1-bounded nets~\cite{EsparzaN94},
(2) using results of~\cite{Peterson81} liveness has the same complexity,
and (3) boundedness is EXPSPACE-complete~\cite{Rackoff78}. 
Furthermore there is a family of bounded nets
$\{\mathcal{N}_n\}_{n\in \mathbb{N}}$ whose size is polynomial in $n$ such that the
size of their reachability set is lower bounded by some Ackermann
function~\cite{Jantzen86}.

\begin{example} 
  An example of marked Petri net is given on Figure~\ref{fig:PN}.
  Petri nets are represented as bipartite graphs where places are
  circles containing their initial number of tokens and transitions
  are rectangles.  When $W^-(p,t)>0$ (resp. $W^+(p,t)>0$) there is an
  edge from $p$ (resp. $t$) to $t$ (resp. $p$) labelled by $W^-(p,t)$
  (resp. $W^+(p,t)$). This label is called the \emph{weight} of
  this edge, and is omitted when its value is equal to
  one.  For sake of readability, one merges edges $p
  \xrightarrow{W^-(p,t)} t$ and $t \xrightarrow{W^+(p,t)} p$ when
  $W^-(p,t)=W^+(p,t)$ leading to a pseudo-edge with two arrows, as in
  the case of $(r_0,t_5)$.

  The net of Figure~\ref{fig:PN} is not live. Indeed $t_0$, $t_1$,
  $t_2$, $t_3$ and $t_4$ will never be enabled due to the absence of
  tokens in $p_0$, $p_1$, $p_2$, $p_{\mathsf{ext}}$.  Suppose that one deletes
  the place $p_{\mathsf{ext}}$ and its input and output edges.  Consider the
  firing sequence $q_3+r_0 \xrightarrow{t_5}q_2+r_0 \xrightarrow{t_6}
  q_1+r_0 \xrightarrow{t_3}p_0+r_0\xrightarrow{t_4}p_2+3q_3+r_0$.
  Since the marking $p_2+3q_3+r_0 $ is (componentwise) larger
     than the initial marking
    $q_3+r_0$, we can iterate this sequence and generate markings with
    an arbitrarily large number of tokens in $p_2$ and $q_3$; this new
    net is unbounded. Applying technics that we will develop in this
    paper (Section~\ref{sec:qualitative}), we will realize that this
    new net is actually live.
\end{example}

\begin{definition}[Stochastic Petri net]
  \label{def:spetrinet}
  A {\em stochastic Petri net} (SPN) is a pair $(\mathcal{N},\lambda)$ where:
\begin{itemize}
\item $\mathcal{N} = (P,T,W^-,W^+)$ is a Petri net;
  \item $\lambda$ is a mapping from $T$ to $\mathbb{R}_{>0}$. 
\end{itemize}
\end{definition}

A marked stochastic Petri net $(\mathcal{N},\lambda,m_0)$ is a stochastic
Petri net equipped with an initial marking.  In a marked stochastic
Petri net, when becoming enabled a transition triggers a random delay
according to an exponential distribution with (firing) rate
$\lambda(t)$.  When several transitions are enabled, a \emph{race}
between them occurs.  Accordingly, given some initial marking $m_0$,
the stochastic process underlying
a SPN is a continuous time Markov chain (CTMC) whose (possibly infinite) set
of states is $\mathcal{R}(m_0)$ and such
that the rate $\mathbf{Q}(m,m')$ of a transition from some $m$ to some
$m'\neq m$ is equal to $\sum_{t:m\xrightarrow{t}m'} \lambda(t)$ (as
usual $\mathbf{Q}(m,m)=-\sum_{m'\neq m}\mathbf{Q}(m,m')$).  Matrix
$\mathbf{Q}$ is called the \emph{infinitesimal generator} of the
CTMC (see~\cite{Cinlar75} for more details).

From a quantitative point of view, one may be interested in studying the
long-run behaviour of the net and in particular in deciding whether
there exists a steady-state distribution and in computing it in the
positive case. When the underlying graph of the CTMC
is strongly connected (i.e. an \emph{irreducible} Markov chain) it
amounts to deciding whether there exists a non-zero distribution $\pi$
over $\mathcal{R}(m_0)$ such that $\pi \cdot \mathbf{Q} =0$.

It is in general non-trivial to decide whether there exists a
  steady-state distribution, and even when such a distribution exists,
  given some state it is hard to compute its steady-state probability
  (see the introduction). 
  Furthermore, even when the net is bounded, the
  size of the reachability set may prevent any feasible computation of
  $\pi$.

Thus one looks for subclasses of nets where the steady-state
distribution $\pi$ can be computed more easily and in particular when
$\pi$ has a \emph{product-form}, that is: there exist a constant
  vector $\mu \in \mathbb{R}_{\ge 0}^P$ only depending on $\mathcal{N}$
  such that for all $m\in
\mathcal{R}(m_0)$, $\pi(m)=G\cdot \prod_{p\in P} \mu_p^{m(p)}$ 
 where $G=\left(\sum_{m\in \mathcal{R}(m_0)}\prod_{p\in P}
  \mu_p^{m(p)} \right)^{-1}$ is the so-called \emph{normalising
  constant}~\cite{Gordon67}.

The most general known class of nets admitting a structural
product-form distribution is the class of
$\mathsf{\Pi}^2$-nets~\cite{HMSS-pe59(4)}. It
 is based on two key ingredients: 
\emph{bags} and \emph{witnesses}.
A bag 
 is a multiset
of tokens that is consumed or produced by some transition.
Considering a bag as a whole, one defines the bag graph whose vertices
are bags and, given a transition $t$, an edge goes from the bag consumed
by $t$ to the bag produced by $t$. Observe that there are at most
$2|T|$ vertices and exactly $|T|$ edges. This alternative
  representation of a net via the bag graph does not lose any
  information:
from the bag graph, one can recover the original net. Formally:

\begin{definition}[Bag graph of a Petri net]
  \label{def:baggraph}
  Let $\mathcal{N} = (P,T,W^-,W^+)$ be a Petri net. Then its bag graph is a
  labelled graph $G_{\mathcal{N}}=(V_{\mathcal{N}},E_{\mathcal{N}})$ defined by:
  \begin{itemize}
  \item $V_{\mathcal{N}}=\{W^-(t),W^+(t) \mid t \in T\}$ the finite set of
    {\em bags};
  \item $E_{\mathcal{N}}=\{W^-(t)\xrightarrow{t} W^+(t) \mid t \in T\}$.
\end{itemize}
\end{definition}

\begin{example} 
  The bag graph of the net of Figure~\ref{fig:PN} is described in
  Figure~\ref{fig:bag-witness}. The bag is written inside the vertex
  (the external label of the vertices will be explained
    later). Observe that this graph has three connected components,
  both of them being strongly connected.
\end{example}

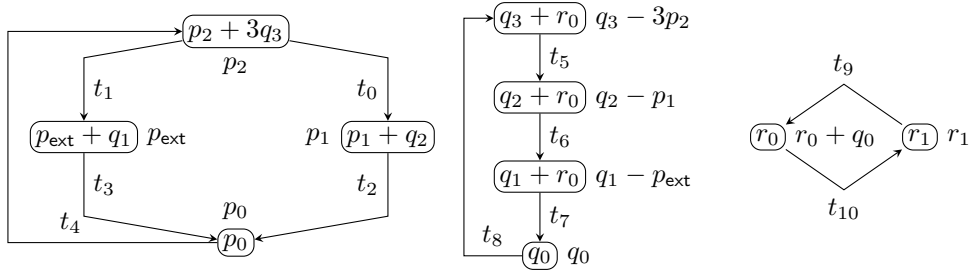
\begin{figure}[htbp]
\begin{center}
\begin{tikzpicture}[xscale=1,yscale=.7]

\path (3,9.75) node[draw,rectangle,rounded corners,inner sep=2pt,minimum size=3mm] (p0) 
[label=below:$p_2$] {$p_2+3q_3$};
\path (5,7.75) node[draw,rectangle,rounded corners,inner sep=2pt,minimum size=3mm] (p1) 
[label=left:$p_1$] {$p_1+q_2$};
\path (1,7.75) node[draw,rectangle,rounded corners,inner sep=2pt,minimum size=3mm] (pext) 
[label=right:$p_{\mathsf{ext}}$] {$p_{\mathsf{ext}}+q_1$};
\path (3,5.75) node[draw,rectangle,rounded corners,inner sep=2pt,minimum size=3mm] (p2) 
[label=above:$p_0$] {$p_0$};

\draw[->,>=stealth] (p0) --(5,9.25)--(p1)node[pos=.5,left] {$t_0$};
\draw[->,>=stealth] (p0) --(1,9.25)-- (pext)node[pos=.5,right] {$t_1$};
\draw[->,>=stealth] (p1) --(5,6.25)node[pos=.5,left] {$t_2$}-- (p2);
\draw[->,>=stealth] (pext) --(1,6.25)node[pos=.5,right] {$t_3$}-- (p2);
\draw[->,>=stealth] (p2) -- (0,5.75)node[pos=.7,above] {$t_4$}--(0,9.75) --(p0);

\path (7,10) node[draw,rectangle,rounded corners,inner sep=2pt,minimum size=3mm] (q0) 
[label= right:$q_3-3p_2$] {$q_3+r_0$};
\path (7,8.5) node[draw,rectangle,rounded corners,inner sep=2pt,minimum size=3mm] (q1) 
[label=right:$q_2-p_1$] {$q_2+r_0$};
\path (7,7) node[draw,rectangle,rounded corners,inner sep=2pt,minimum size=3mm] (q2) 
[label=right:$q_1-p_{\mathsf{ext}}$] {$q_1+r_0$};
\path (7,5.5) node[draw,rectangle,rounded corners,inner sep=2pt,minimum size=3mm] (q3) 
[label=right:$q_0$] {$q_0$};

\draw[->,>=stealth] (q0) -- (q1)node[pos=.5,right] {$t_5$};
\draw[->,>=stealth] (q1) -- (q2)node[pos=.5,right] {$t_6$};
\draw[->,>=stealth] (q2) -- (q3)node[pos=.5,right] {$t_7$};
\draw[->,>=stealth] (q3) -- (6,5.5)node[pos=.5,above] {$t_8$} -- (6,10)--(q0);

\path (10,7.75) node[draw,rectangle,rounded corners,inner sep=2pt,minimum size=3mm] (r0) 
[label=right:$r_0+q_0$] {$r_0$};
\path (12,7.75) node[draw,rectangle,rounded corners,inner sep=2pt,minimum size=3mm] (r1) 
[label=right:$r_1$] {$r_1$};

\draw[->,>=stealth] (r0) -- (11,6.75)--(r1)node[pos=0,below] {$t_{10}$};
\draw[->,>=stealth] (r1) -- (11,8.75)--(r0)node[pos=0,above] {$t_{9}$};

\end{tikzpicture}
\caption{Bags and witnesses.}
\label{fig:bag-witness}
\end{center}
\end{figure}

We now turn to the notion of witness. 
A queuing network models a single activity flow where activities are
modelled by queues and clients leave their current queue when served
and enter a new one depending on a routing probability.  In
$\mathsf{\Pi}^2$-nets there are several activity flows, one per component of
the bag graph.  So one wants to witness production and consumption of
every bag $b$ by the transition firings.
In order to
witness it,
one looks for a linear combination of the places $\mathit{wit}$ such that for
every firing of a transition $t$ that produces (resp. consumes) the
bag $b$, for every marking $m$, $m \cdot \mathit{wit}$ is increased
(resp. decreased) by one unit, and such that all other transition
firings let $m \cdot \mathit{wit}$ invariant.

\begin{definition}[Witness of a bag]
  \label{def:witness}
  Let $\mathcal{N}=(P,T,W^-,W^+)$ be a a Petri net, $b\in V_{\mathcal{N}}$ and $wit
  \in \mathbb{Q}^P$.  Then $wit$ is a witness of $b$ if:
  \[
\begin{cases}
\mathit{wit}\cdot W(t) =-1 & \mbox{if } W^{-}(t) = b;\\
\mathit{wit}\cdot W(t) =1 & \mbox{if } W^{+}(t)= b;\\
\mathit{wit}\cdot W(t) =0 & \mbox{otherwise}.
\end{cases}
\]
\end{definition}

\begin{example}
  All bags of the net of Figure~\ref{fig:PN} have (non unique)
  witnesses.  We have depicted them close to their vertices in
  Figure~\ref{fig:bag-witness}.  For instance, consider the bag
  $b=q_1+r_0$: transition $t_6$ produces $b$ while transition $t_7$
  consumes it. Let us check that $w=q_1-p_{\mathsf{ext}}$ is a witness of
  $b$. $t_6$ produces a token in $q_1$ and the marking of $p_{\mathsf{ext}}$
  is unchanged.  $t_7$ consumes a token in $q_1$ and the marking of
  $p_{\mathsf{ext}}$ is unchanged. The other transitions that change the
  marking of $q_1$ and $p_{\mathsf{ext}}$ are $t_1$ and $t_3$. However since
  they simultaneously produce or consume a token in both places, $m
  \cdot w$ is unchanged ($m$ is the current marking).
\end{example}

The definition of $\mathsf{\Pi}^2$-nets
relies
on structural properties of the net and on the existence of
witnesses.  Every connected component of the graph bag will
represent an activity flow of some set of processes
where every activity (i.e. a bag) has a witness. 

\begin{definition}[$\mathsf{\Pi}^2$-net]
  \label{def:pideuxnet}
  Let $\mathcal{N}$ be a Petri net. Then $\mathcal{N}$ is a $\mathsf{\Pi}^2$-net if:
  \begin{itemize}
   \item every connected 
  component of $G_{\mathcal{N}}$ is strongly connected;
   \item every bag $b$ of $V_{\mathcal{N}}$
   admits a witness (denoted $\mathit{wit}_b$).
  \end{itemize}
\end{definition}

Observe that the first
condition called \emph{weak reversibility}
ensures that the reachability graph is strongly connected
since the firing of any
transition $t$ can be ``withdrawn'' by the firing of transitions
occurring along a path from $W^+(t)$ to $W^-(t)$ in the bag graph.
The complexity of reachability in weakly reversible nets
is still high: EXPSPACE-complete~\cite{CardozaLM76}.

The next theorem shows the interest of $\mathsf{\Pi}^2$-nets.  Let us define
$\lambda(b)$ the \emph{firing rate} of a bag $b$ by
$\lambda(b)=\sum_{t \mid W^-(t)=b}\lambda(t)$ and the choice
probability $pr_t$ of transition $t$ by
$pr_t=\frac{\lambda(t)}{\lambda({W^-(t)})}$.  The routing matrix
$\mathbf{P}$ of bags is the stochastic matrix indexed by bags such
that for all $t$, $\mathbf{P}(W^-(t),W^+(t))=pr_t$ and
$\mathbf{P}(b,b')=0$ otherwise.  Consider $\mathbf{vis}$ some positive
solution of $\mathbf{vis}\cdot \mathbf{P}=\mathbf{vis}$. Since
$\mathbf{P}$ is a stochastic matrix such a
  vector always exists but
is not unique in general; however given $b,b'$ two bags of the same
connected component, $\frac{\mathbf{vis}(b')}{\mathbf{vis}(b)}$
measures the ratio between visits of $b'$ and $b$ in the discrete time
Markov chain induced by $\mathbf{P}$. 

\begin{theorem}[\cite{HMSS-pe59(4)}]
 \label{th:pideux}
Let $(\mathcal{N},\lambda,m_0)$ be a marked stochastic $\mathsf{\Pi}^2$-net. Then defining 
for all $m \in \mathcal{R}(m_0)$,
$\mathbf{v}(m)=\prod_{b\in V_{\mathcal{N}}} 
\left(\frac{\mathbf{vis}(b)}{\lambda(b)}\right)^{m \cdot \mathit{wit}_b }$,
the next assertions hold:
\begin{itemize}
 \item  $\mathbf{v} \cdot \mathbf{Q}=0$;
 \item $(\mathcal{N},\lambda,m_0)$ is ergodic iff $\|\mathbf{v}\| <\infty$ 
 \item when ergodic, the associated Markov chain
 admits $\|\mathbf{v}\|^{-1}\mathbf{v}$ 
 as steady-state distribution.
\end{itemize}
\end{theorem}

Let us discuss the computational complexity of the product-form of the
previous theorem. First deciding whether a net is a $\mathsf{\Pi}^2$-net is
straightforwardly performed in polynomial time.  The computation of
the visit ratios, the witnesses and the rate of bags can also be done 
in polynomial time.  So computing an item of vector $\mathbf{v}$ is
easy. However without additional restriction on the nets the
normalising constant $\|\mathbf{v}\|^{-1}$ requires to enumerate
all items of $\mathcal{R}_{\mathcal{N}}(m_0)$,
which can be prohibitive.

So in~\cite{HMN-fi13}, the authors introduce $\mathsf{\Pi}^3$-net, a subclass
of $\mathsf{\Pi}^2$-net which still strictly generalises closed queuing networks, 
obtained by structuring the activity flows of
the net represented by the components of the bag graph.  First there
is a bijection between places and bags such that the input
(resp. output) transitions of the bag produce (consume) one token of
this place. The other places occurring in the bag may be viewed as
resources associated with the bag and thus the \emph{potential} of the
bag is its total number of resources.  Second the components of the
graph may be ordered as $N$ layers such that the resources of a bag
occurring in layer $i$ correspond to places associated with bags of
layer $i-1$ (for $i>1$) and more precisely to those with maximal
potential. Informally a token in such a place means that it is a
resource available for the upper layer.

\begin{definition}[$\mathsf{\Pi}^3$-net]
  \label{def:pitrois}
  Let $\mathcal{N}$ be a net. Then $\mathcal{N}$ is an $N$-closed $\mathsf{\Pi}^3$-net
  if:
\begin{itemize}
  \item There is a bijection between $P$ and $V_{\mathcal{N}}$.
  Denoting $b_p$ the bag associated with place $p$ (and
  $p_b$ the place associated with bag $b$), we have $b_p(p)=1$.\\
  The potential of a place $\mathsf{pot}(p)$ is equal to $\|b_p\|-1$.
  \item $V_{\mathcal{N}}$ is partitioned into $N$ strongly connected
  components $V_1,\ldots,V_N$. 
  One denotes:\\ $P_i=\{p_b \mid b\in V_i\}$ and
  $P_i^{\max}=\mathrm{argmax}(\mathsf{pot}(p) \mid p \in P_i)$. By convention, 
  $P_0^{\max}=P_0=\emptyset$.
  \item For all $b \in V_i$ and $p\in P\setminus \{p_b\}$, 
  $b(p)>0$ implies $p\in P_{i-1}^{\max}$.
\end{itemize}
A net is an $N$-open $\mathsf{\Pi}^3$-net if it is obtained by deleting some
place $p_{\mathsf{ext}} \in P_N$ (and its input/output edges) from an
$N$-closed $\mathsf{\Pi}^3$-net.
\end{definition}

Given an open $\mathsf{\Pi}^3$-net $\mathcal{N}$, $\overline{\mathcal{N}}$ denotes the
closed net based on which $\mathcal{N}$ has been defined. For every $1 \le i
\le N$, we will later write $P_i^{\neg \max}$ for the set $P_i
\setminus P_i^{\max}$, and $T_i$ for the set of transitions
$t$ such that $W^-(t) \in V_i$.  The next proposition establishes that
$\mathsf{\Pi}^3$-nets are product-form Petri nets.

\begin{restatable}{proposition}{proppideuxtrois}
\label{proppideuxtrois}
  Let $\mathcal{N}$ be a (closed or open) $\mathsf{\Pi}^3$-net.  Then $\mathcal{N}$ is a
  $\mathsf{\Pi}^2$-net.
\end{restatable}

\begin{example}
  The net of Figure~\ref{fig:PN} is a 3-closed $\mathsf{\Pi}^3$-net.  We have
  used different colors for the layers: $P_3=\{p_0,p_1,p_2,p_{\mathsf{ext}}\}$
  (in red), $P_2=\{q_0,q_1,q_2,q_3\}$ (in green) and $P_1=\{r_0,r_1\}$
  (in blue). The place $q_0$ (resp. $q_1,q_2,q_3$) is associated with
  the bag $q_0$ (resp.  $q_1+r_0,q_2+r_0,q_3+r_0$) and has potential
  $0$ (resp. $1,1,1$).  Thus $P_2^{\max}=\{q_1,q_2,q_3\}$ and indeed
  $q_0$ does not occur in bags of layer $3$. In order to highlight the
  use of resources, edges between places of $P_i$ and transitions of
  $T_{i+1}$ are depicted in gray.
\end{example}

The next theorem shows the interest of closed $\mathsf{\Pi}^3$-nets.
\begin{theorem}[\cite{HMN-fi13}]
\label{th:pitrois}
Let $(\mathcal{N},\lambda,m_0)$ be a $N$-closed $\mathsf{\Pi}^3$-net.  Then:
\begin{itemize}
 \item $(\mathcal{N},m_0)$ is bounded.
 \item One can decide whether $(\mathcal{N},m_0)$ is live in polynomial time.
 \item When $(\mathcal{N},m_0)$ is live, one can decide in polynomial time, 
 given $m$, whether $m\in \mathcal{R}_\mathcal{N}(m_0)$.
\item For any fixed $N$, when $(\mathcal{N},m_0)$ is live, one can compute
  the normalising constant of the steady-state distribution
  (i.e. $\|\mathbf{v}\|^{-1}$ in Theorem~\ref{th:pideux}) in
  polynomial time with respect to $|P|$, $|T|$, the maximal weight of
  the net's edges, and $\|m_0\|$ (thus in pseudo-polynomial time
  w.r.t. the size of $m_0$).
\end{itemize}
\end{theorem}

The above results do not apply to infinite-state systems and in
particular to the systems generated by open $\mathsf{\Pi}^3$-nets.
In addition, the polynomial-time complexity for computing the
normalising constant requires to fix $N$. We address these issues in
the next sections.

\section{Qualitative analysis}
\label{sec:qualitative}

In this section we first give a simple characterisation of the
liveness property in a marked $\mathsf{\Pi}^3$-net. We then fully characterise
the set of reachable markings in a live marked $\mathsf{\Pi}^3$-net.
These characterisations give polynomial-time algorithms for deciding
liveness of a marked $\mathsf{\Pi}^3$-net, and the boundedness property of a
live marked $\mathsf{\Pi}^3$-net. We end the section
with a coNP-hardness result for the boundedness property of a marked
$\mathsf{\Pi}^3$-net, when it is not live.
For the rest of this section, we assume $\mathcal{N}=(P,T,W^-,W^+)$ is an
open or closed $\mathsf{\Pi}^3$-net with $N$ layers. We further use the
notations of Definition~\ref{def:pitrois}. In particular, if $\mathcal{N}$ is
open, then we write $p_\mathsf{ext}$ for the place which has been removed
(and we call it virtual). We therefore set $P_N^\star = P_N \cup
\{p_\mathsf{ext}\}$ if the net is open and $P_N^\star = P_N$ otherwise; For
every $1 \le i \le N-1$ we define $P_i^\star = P_i$; And we set
$P^\star = \bigcup_{i=1}^N P_i^\star$.

\subsection{Liveness analysis}

We give a simple characterisation of the liveness property through a
dependence between the number of tokens at some layer and potentials
of places activated on the next layer.  More precisely, for every $1
\le i \le N-1$, $\mathsf{Live}_i$ is defined as the set of markings $m$ such
that:
\[ 
m \cdot P_i \ge \min\{\mathsf{pot}(p) \mid p \in P^\star_{i+1} \text{ and }
(m(p)>0 \text{ or } p=p_{\mathsf{ext}})\}.
\]
Note that $p = p_\mathsf{ext}$ can only happen when $\mathcal{N}$ is open and
$i=N-1$. We additionally define $\mathsf{Live}_N$ as the set of markings $m$
such that $m \cdot P_N >0$ if $\mathcal{N}$ is closed, and as the set of all
markings if $\mathcal{N}$ is open. 

For every $1 \le i \le N$ (resp. $1 \le i \le N-1$), we define
$\mathsf{POT}_i = \max\{\mathsf{pot}(p) \mid p \in P_i\}$ when $\mathcal{N}$ is closed
(resp. open). When $\mathcal{N}$ is open, we write $\mathsf{POT}_N = \mathsf{pot}(p_\mathsf{ext})$.
Given a marking $m$,
when no place $p$ fulfills $p \in P^\star_{i+1} \text{ and }
(m(p)>0 \text{ or } p=p_{\mathsf{ext}})$, the minimum is equal to $\mathsf{POT}_{i+1}$. 
Thus given a marking
$m$,  the condition $m\in \mathsf{Live}_i$ for $i<n$ only depends on the values
of $m(p)$ for $p\in P_i \cup P_{i+1}^{\neg \max}$. 

The intuition behind condition $\mathsf{Live}_i$ is the following:
  transitions in $\bigcup_{j \le i} T_j$ cannot create new tokens on
  layer $i$ (layer $i$ behaves like a state machine, and smaller
  layers do not change the number of tokens in that layer); therefore,
  to activate a transition of $T_{i+1}$ out of some marked place $p
  \in P_{i+1}$, it must be the case that enough tokens are already
  present on layer $i$; hence there should be at least as many tokens
  in layer $i$ as the minimal potential of a marked place in layer
  $i+1$.  When $\mathcal{N}$ is open, the virtual place $p_\mathsf{ext}$ behaves like
  a source of tokens, hence it is somehow always ``marked''; this is
  why it is taken into account in the right part of $\mathsf{Live}_i$.
The following characterisation
was already stated
in~\cite{HMN-fi13} in the restricted case of closed nets.

\begin{theorem}\label{theoliveness}
  \label{theo:liveness}
  A marking $m$ is live if and only if for every $1 \le i \le N$, $m
  \in \mathsf{Live}_i$.
\end{theorem}
\begin{example}
Building on the marked Petri net of
    Figure~\ref{fig:PN}, the marking $q_3+r_0$ is live when the net is
    open but not live if the net is closed. Indeed, transitions of the
    two first layers can trivially be activated from $q_3+r_0$ (hence
    by weak-reversibility from every reachable marking). We see that
    in the case of the closed net, transitions of layer $3$ cannot be
    activated (no fresh token can be produced on that layer). On the
    contrary, in the case of the open net, the token in $q_3$ can be
    moved to $q_1$, which will activate transition $t_3$; from there,
    all transitions of layer $3$ will be eventually activated.
\end{example}

As a consequence of the characterisation of
Theorem~\ref{theo:liveness}, we get:

\begin{corollary}
  We can decide the liveness of a marked $\mathsf{\Pi}^3$-net in polynomial
  time.
\end{corollary}

\subsection{Reachable markings}

We will now give a characterisation of the set of reachable markings
$\mathcal{R}_\mathcal{N}(m_0)$ when $m_0$ is live. We will first give linear
invariants of the net: those are vectors in the left kernel of $W$ (or
$P$-flows). The name ``invariants'' comes from the fact that they will
allow to infer real invariants satisfied by the reachable markings.
Furthermore, for every $1 \le i \le N$, for every $p \in P_i$, we
define $\mathsf{cin}(p) = \mathsf{POT}_i - \mathsf{pot}(p)$. Except when $\mathcal{N}$ is open and
$i=N$, the $\mathsf{cin}$ value of a place is nonnegative.
\begin{restatable}{proposition}{propinvariants}
  \label{propinvariants}
  \label{prop:invariants}
  The following vectors are linear
  invariants of $\mathcal{N}$:
  \begin{itemize}
  \item for every $1 \leqslant i \leqslant N-1$, $v^{(i)} = \sum_{p \in P_i} p +
    \sum_{p \in P_{i+1}} \mathsf{cin}(p)\ p$;
  \item if $\mathcal{N}$ is closed, $v^{(N)} = \sum_{p \in P_N} p$.
  \end{itemize}
\end{restatable}

\begin{floatingfigure}{4cm}
  \begin{tikzpicture}[xscale=.5,yscale=.4]
    \path [use as bounding box] (-1.5,-4.5) -- (3,3);
\path (0,3) node[draw,circle,inner sep=2pt,minimum size=4mm] (p0) 
[label=left:$p_2:3$] {};
\path (0,-3) node[draw,circle,inner sep=2pt,minimum size=4mm] (pext) 
[label=left:$p_{\mathsf{ext}}:1$] {};
\path (0,0) node[draw,rectangle,inner sep=2pt,minimum size=4mm] (t1) 
[label=left:$t_1$] {};
\draw[->,>=stealth] (p0) -- (t1);
\draw[->,>=stealth] (t1) -- (pext);
\path (3,3) node[draw,circle,inner sep=2pt,minimum size=4mm] (q0) 
[label=right:$q_3:1$] {};
\path (3,-3) node[draw,circle,inner sep=2pt,minimum size=4mm] (q2) 
[label=right:$q_1:1$] {};
 \draw[->,>=stealth] (t1) -- (q2);
 \draw[->,>=stealth] (q0) -- (t1) node[pos=.3,below]{$3$};
\draw [dotted] (1.5,-4.5) -- (1.5,3.5);
\path (0,-4.5) node {Layer $3$};
\path (3,-4.5) node {Layer $2$};
\end{tikzpicture}
\end{floatingfigure}

\smallskip
First observe that for $i<N-1$, 
$\mathsf{Supp}(v^{(i)} )=P_i \cup P_{i+1}^{\neg \max}$.
Thus given a marking $m$,
only firing of transitions $t \in T_{i} \cup T_{i+1}$ could change 
$m \cdot v^{(i)} $. This is not the case of a transition $t \in T_i$
since it moves a token from a place of $P_i$ to another one. 
To give an intuition why transitions in $T_{i+1}$ do not change
$m \cdot v^{(i)}$, we consider part of
  the closed net (that is, $p_\mathsf{ext}$ is a real place) of
  Figure~\ref{fig:PN} depicted on the right, where numbers close to
  place names are potential values. We focus on transition $t_1$ and
  explain why $m \cdot v^{(2)}$ is unchanged by its firing. The impact of
  transition $t_1$ is to decrease the sum $\sum_{p \in P_2} p$ by $2$;
  due to the weights of places of $P_3$ in $v^{(2)}$, place $p_2$
  counts as $0$ and place $p_\mathsf{ext}$ counts as $+2$. 
  This intuition extends into a
  formal proof.

\medskip For every $1 \le i \le N-1$, we set $\mathbf{C}_i^{m_0} = m_0 \cdot
v^{(i)}$,
and if $\mathcal{N}$ is closed, we set $\mathbf{C}_N^{m_0} = m_0 \cdot v^{(N)}$.
For every $1 \le i \le N-1$, we define $\mathsf{Inv}_i(m_0)$ as the set of
markings $m$ such that $m \cdot v^{(i)} = \mathbf{C}_i^{m_0}$, and if
$\mathcal{N}$ is closed, we define $\mathsf{Inv}_N(m_0)$ as the set of markings $m$
such that $m \cdot v^{(N)} = \mathbf{C}_N^{m_0}$.  For uniformity, if
$\mathcal{N}$ is open, we define $\mathsf{Inv}_N(m_0)$ as the set of all markings.
As a consequence of Proposition~\ref{prop:invariants}, we get:
\begin{corollary}
  \label{coro:invariants}
  ${\displaystyle \mathcal{R}_{\mathcal{N}}(m_0) \subseteq \bigcap_{i=1}^N \mathsf{Inv}_i(m_0)}$.
\end{corollary}

More complex invariants were given in~\cite{HMN-fi13} for closed
$\mathsf{\Pi}^3$-nets. The advantage of the above invariants is that each of
them only involves two neighbouring layers. This will have a huge
impact on various complexities, and will allow the development of
our methods for quantitative analysis.
\begin{example}
  Going back to the Petri net of Figure~\ref{fig:PN}, with initial
  marking $m_0 = q_3+r_0$. We first consider the closed net. Then:
  $\mathsf{POT}_3 = 2$ and $\mathsf{POT}_2 = 1$. Therefore:
  \[
  \left\{\begin{array}{l} 
      \mathsf{Inv}_3(m_0) =  \{m \mid \sum_{i=0}^2 m(p_i) + m(p_\mathsf{ext}) = 0\} \\[0.5mm]
      \mathsf{Inv}_2(m_0) =  \{m \mid \sum_{i=0}^3 m(q_i) + 3 m(p_0) + 2 m(p_1) + 2 m(p_\mathsf{ext}) = m_0(q_3) = 1\} \\[0.5mm]
      \mathsf{Inv}_1(m_0) =   \{m \mid m(r_0)+m(r_1)+m(q_0) = m_0(r_0) = 1\}
    \end{array}\right.
  \]
  We now turn to the open net, obtained by deleting
  $p_\mathsf{ext}$. Definition of $\mathsf{POT}_3$ differs from the previous case:
  $\mathsf{POT}_3 = \mathsf{pot}(p_\mathsf{ext}) = 1$ and $\mathsf{POT}_2 = 1$. Therefore:
  \[
  \left\{\begin{array}{l} 
      \mathsf{Inv}_3(m_0) = \mathbb{N}^P \\[0.5mm]
      \mathsf{Inv}_2(m_0) = \{m \mid \sum_{i=0}^3 m(q_i)
      + m(p_0) - 2 m(p_2) = m_0(q_3) = 1\} \\[0.5mm]
      \mathsf{Inv}_1(m_0) = \{m \mid  m(r_0)+m(r_1)+m(q_0) = m_0(r_0) = 1\}
    \end{array}\right.
  \]
\end{example}
The invariants of Corollary~\ref{coro:invariants} do not fully
characterise the set of reachable markings, since they do not take
into account the enabling conditions of the transitions.  However,
they will be very helpful for characterising the reachable markings
when $m_0$ is live.
\begin{theorem}\label{theoreach}
  \label{thm:CNS-reach}
  \label{theo:reach}

  Suppose that $(\mathcal{N},m_0)$ is a live $\mathsf{\Pi}^3$-net. Then:
  \[
  \mathcal{R}_\mathcal{N}(m_0) = \bigcap_{i=1}^N \mathsf{Inv}_i(m_0)~ \cap ~\bigcap_{i=1}^N \mathsf{Live}_i
  \]
  Thus reachability in live $\mathsf{\Pi}^3$-nets can be checked in polynomial time.
\end{theorem}

\begin{example}
  In the open Petri net of Figure~\ref{fig:PN}, with initial
  marking $m_0 = q_3+r_0$, the sets $\mathsf{Live}_i$ are
  $\mathsf{Live}_1 = \{m \mid m(r_0)+m(r_1)+m(q_0) \geqslant 1\}$,
  $\mathsf{Live}_2 = \{m \mid \sum_{i=0}^3 m(q_i) + m(p_0) \geqslant 1\}$ and
  $\mathsf{Live}_3 = \mathbb{N}^P$.
  Observing that $\mathsf{Inv}_i(m_0) \subseteq \mathsf{Live}_i$ for $1 \leqslant i \leqslant 3$,
  the net has reachability set
  \[\mathcal{R}_\mathcal{N}(m_0) = \{m \mid m(r_0)+m(r_1)+m(q_0) = 1 \text{ and } {\textstyle\sum_{i=0}^3} m(q_i) + m(p_0) - 2 m(p_2) = 1\}.\]

\end{example}

The idea of the proof when the net is closed is to show that,
  from every marking $m$ satisfying the right handside condition in
  the theorem, one can reach a specific marking $m_0^\ast$ (where, for
  every $1 \le i \le N$, $\mathbf{C}_i^{m_0}$ tokens are in one arbitrary
  place of $P_i^{\max}$).  Hence, given two markings $m$ and $m'$
  satisfying the conditions, $m \to^* m_0^\ast$ and $m' \to^*
  m_0^\ast$, which implies by weak reversibility of the net: $m \to^*
  m_0^\ast \to^* m'$. This is in particular the case from $m_0$: every
  marking satisfying the conditions is reachable from $m_0$. In the
  case of an open net, this is a bit more tricky, and a joint marking
  to every pair $(m,m')$ of markings satisfying the conditions has to
  be chosen.

\subsection{Boundedness analysis}

As a consequence of the characterisation given in
Theorem~\ref{theo:reach}, we get:
\begin{restatable}{corollary}{coroboundness}
\label{coroboundness}
  \label{coro:boundedness}
  We can decide the boundedness of a live marked $\mathsf{\Pi}^3$-net in
  polynomial time. 
\end{restatable}
Indeed, it can be shown
that, if $\mathcal{N}$ is closed, then $\mathcal{N}$
is bounded, and that if $\mathcal{N}$ is open, then it is bounded if and only
if $\mathsf{cin}(q) > 0$ for all $q \in P_N$ (that is, $p_\mathsf{ext}$ has maximal
potential, and no other place of $P_N$ has maximal potential).
Furthermore, if $\mathcal{N}$ is bounded, then the overall number of
  tokens in the net is bounded by $\sum_{i=1}^{N-1} \mathbf{C}_i^{m_0}$
  (resp. $\sum_{i=1}^N \mathbf{C}_i^{m_0}$) if the net is open
  (resp. closed).

The polynomial-time complexity of Corollary~\ref{coro:boundedness} is
in contrast with the following hardness result, which can be obtained
by a reduction from the independent-set problem.
\begin{restatable}{proposition}{coNPhardness}
\label{coNPhardness}
  Deciding the boundedness of a marked $\mathsf{\Pi}^3$-net which is not live is
  coNP-hard. The reachability (and even the coverability) problem is
  NP-hard.
\end{restatable}

\section{Quantitative analysis}
\label{sec:quantitative}

Contrary to closed $\mathsf{\Pi}^3$-nets, open $\mathsf{\Pi}^3$-nets may not be
  ergodic. In this section, we first give a simple characterisation of
  the ergodicity property for open $\mathsf{\Pi}^3$-nets, which gives us a
  polynomial-time algorithm for deciding ergodicity. We then provide a
  polynomial-time algorithm for computing the steady-state
  distribution of ergodic (open and closed) $\mathsf{\Pi}^3$-nets.

For the rest of this section, we assume that $(\mathcal{N},\lambda,m_0)$ is
a stochastic $\mathsf{\Pi}^3$-net with $N$ layers, and that $m_0$ is live. Let $\mathbf{W}$ be
the maximal weight of the edges of $\mathcal{N}$.  Then we assume that the
constants $\mu_p = \prod_{b \in \mathcal{V}_{\mathcal{N}}} \left(\mathbf{vis}(b) /
  \lambda_b\right)^{wit_b \cdot p}$ have already been precomputed (in
polynomial time with respect to $|P|$, $|T|$ and
$\log(1+\mathbf{W})$).~\footnote{In the rest of this section, we will design polynomial-time procedures  w.r.t. $\mathbf{W}$, hence
    polynomial if it is encoded in unary and pseudo-polynomial if it
    is encoded in binary.}

In what follows, and for every vector $\delta \in \mathbb{N}^P$, we denote
by $\hat{v}(\delta)$ the product $\prod_{p \in P} \mu_p^{\delta(p)}$.
Consequently, the vector $\mathbf{v}$ mentioned in Theorem~\ref{th:pideux}
is then defined by $\mathbf{v}(m) = \hat{v}(m)$ for all markings $m \in
\mathcal{R}(m_0)$, and its norm is $\|\mathbf{v}\| = \sum_{m \in \mathcal{R}(m_0)}
\hat{v}(m)$. Note that $\mathbf{v}$ and $\hat{v}$ only differ by
  their domain. In addition, in what follows, and for every set $Z
\subseteq P$, we simply denote by $\mathsf{cin}(Z)$ the formal sum $\sum_{p
  \in Z} \mathsf{cin}(p) p$.

\subsection{Ergodicity analysis}
\label{section:ergodic}

We assume here that $\mathcal{N}$ is open. We give a simple characterisation
of the ergodicity property through a comparison of the constants
$\mu_p$ for a limited number of places $p$. Those constraints
  express congestion situations that may arise; we show that they are
  sufficient. These places are the elements of the subset $Y$ of
places that is defined by $Y = P_N \cup P_{N-1}^{\max}$.  In
particular, as soon as the initial marking $m_0$ is live, then the
ergodicity of the stochastic net $(\mathcal{N},\lambda,m_0)$ does \emph{not}
depend on $m_0$.

According to Theorem~\ref{th:pideux}, the net is ergodic if and only
if the norm $\|\mathbf{v}\| = \sum_{m \in \mathcal{R}(m_0)} \hat{v}(m)$ is finite.
Hence, deciding ergodicity amounts to deciding the convergence of a
sum.  The following characterisation
holds.

\begin{restatable}{theorem}{theoergodic}
\label{theoergodic}
  \label{theo:ergodic}
Let $(\mathcal{N},\lambda,m_0)$ be a live open stochastic $\mathsf{\Pi}^3$-net with $N$ layers. 
This net is ergodic if and only if all of the following inequalities hold:
\begin{itemize}
 \item for all places $p \in P_N$, if $\mathsf{cin}(p) = 0$, then $\mu_p < 1$;
 \item for all places $p, q \in P_N$, if $\mathsf{cin}(p) > 0 > \mathsf{cin}(q)$, then $\mu_p^{|\mathsf{cin}(q)|} \mu_q^{\mathsf{cin}(p)} < 1$;
 \item for all places $p \in P_{N-1}^{\max}$ and $q \in P_N$, if $0 > \mathsf{cin}(q)$, then $\mu_p^{|\mathsf{cin}(q)|} \mu_q < 1$.
\end{itemize}
\end{restatable}

\begin{proof}[Proof (sketch)]
Let $\mathcal{F}$ be the family formed of the vectors $p$ (for $p \in P_N$ such that $\mathsf{cin}(p) = 0$),
$\mathsf{cin}(p) q - \mathsf{cin}(q) p$ (for $p, q \in P_N$ such that $\mathsf{cin}(p) > 0 > \mathsf{cin}(q)$) and
$q - \mathsf{cin}(q) p$ (for $p \in P_{N-1}^{\max}$ and $q \in P_N$ such that $0 > \mathsf{cin}(q)$).
Let also $\mathcal{L}$ be the sublattice of $\mathbb{N}^P$ generated by the vectors in $\mathcal{F}$,
and let $\mathcal{G}$ be the finite subset of $\mathbb{N}^P$ formed of those vectors whose entries
are not greater than some adequately chosen constant $\mathbf{G}$.

Since $m_0$ is live, Theorem~\ref{theo:reach} applies, which allows us
to prove the inclusions $\{m_0\} + \mathcal{L} \subseteq \mathcal{R}(m_0) \subseteq \mathcal{G} + \mathcal{L}$.
Hence, the sum $\sum_{m \in \mathcal{R}(m_0)} \hat{v}(m)$ is
finite iff the sum $\sum_{m \in \mathcal{L}} \hat{v}(m)$ is finite, i.e.
iff each constant $\hat{v}(\delta)$ is (strictly) smaller than $1$, for $\delta \in \mathcal{F}$.
\end{proof}

\begin{example}
\label{ex:ergodic}
Going back to the open $\mathsf{\Pi}^3$-net of Figure~\ref{fig:PN}, with any
live initial marking $m_0$, we obtain the following necessary and
sufficient conditions for being ergodic:
  \[\mu_{p_1} < 1 \text{, } \mu_{p_0}^2 \mu_{p_2} < 1 \text{, } \mu_{q_1}^2 \mu_{p_2} < 1
  \text{, } \mu_{q_2}^2 \mu_{p_2} < 1 \text{ and } \mu_{q_3}^2 \mu_{p_2} < 1.\]
\end{example}

As a consequence of the characterisation of
Theorem~\ref{theo:ergodic}, we get:

\begin{corollary}
  We can decide the ergodicity of a marked, live and open stochastic
  $\mathsf{\Pi}^3$-net in polynomial time.
\end{corollary}

\subsection{Computing the steady-state distribution}
\label{sec:computing}

In case the $\mathsf{\Pi}^3$-net is ergodic, it remains to compute its
steady-state distribution, given by $\pi(m) = \|\mathbf{v}\|^{-1} \mathbf{v}(m)$
for all $m \in \mathcal{R}(m_0)$.  Since we already computed $\mathcal{R}(m_0)$
and $\mathbf{v}(m)$ for all markings $m \in \mathcal{R}(m_0)$, it remains to
compute the normalising constant $\|\mathbf{v}\|$.

This section is devoted to proving the following result.

\begin{theorem}\label{theconstant}
  \label{theo:constant}
  Let $(\mathcal{N},\lambda,m_0)$ be a live ergodic stochastic $\mathsf{\Pi}^3$-net. 
  There exists an algorithm for
  computing the normalising constant $\|\mathbf{v}\|$ in polynomial time
  with respect to $|P|$, $|T|$, $\mathbf{W}$, and $\|m_0\|$ (thus in
    pseudo-polynomial time).
\end{theorem}

This theorem applies to both closed and open $\mathsf{\Pi}^3$-nets.  For
closed nets, it provides a similar yet stronger result than
Theorem~\ref{th:pitrois}, where the polynomial-time complexity was
obtained only for a fixed value of $N$ (the number of layers of the
net).

We prove below Theorem~\ref{theo:constant} in the case of open nets.
The case of closed nets is arguably easier: one can transform a closed
net into an equivalent open net by adding a layer $N+1$ with one place
(and one virtual place $p_\mathsf{ext}$), and set a firing rate $\lambda_t =
0$ for all transitions $t$ of the layer $N+1$. We therefore
  assume for the rest of this section that $\mathcal{N}$ is open.

\medskip We first describe a naive approach. The normalisation
constant $\|\mathbf{v}\|$ can be computed as follows.  Recall the family
$\mathcal{F}$ introduced in the proof of Theorem~\ref{theo:ergodic}.  We may
prove that the set $\mathcal{R}(m_0)$ is a union of (exponentially many)
translated copies of the lattice $\mathcal{L}$ generated by $\mathcal{F}$.  These
copies may intersect each other, yet their intersections themselves
are translated copies of $\mathcal{L}$.  Hence, using an inclusion-exclusion
formula and a doubly exponential computation step, computing the sum
$\|\mathbf{v}\| = \sum_{\mathfrak{m} \in \mathcal{R}(m_0)} \hat{v}(\mathfrak{m})$
reduces to computing the sum $\sum_{\ell \in \mathcal{L}} \hat{v}(\ell)$.

The family $\mathcal{F}$ is not free a priori, hence computing this latter
sum is not itself immediate.  Using again inclusion-exclusion
formul\ae, we may write $\mathcal{L}$ as a finite, disjoint union of
exponentially many lattices generated by free subfamilies of
$\mathcal{F}$. This last step allows us to compute $\sum_{\ell \in \mathcal{L}}
\hat{v}(\ell)$, and therefore $\|\mathbf{v}\|$.

Such an approach suffers from a prohibitive computational cost. Yet it
is conceptually simple, and it allows proving rather easily that
$\|\mathbf{v}\|$ is a rational fraction in the constants $\mu_p$, whose
denominator is the product $\prod_{\ell \in \mathcal{F}}(1-\hat{v}(\ell))$.

\begin{example}
  Going back to the open $\mathsf{\Pi}^3$-net of Figure~\ref{fig:PN}, with initial
  marking $m_0 = q_3+r_0$, this algorithm allows us to compute
  \begin{align*}
  \hspace{-2mm}\|\mathbf{v}\| =~ & \frac{\mu_{q_0} a + (\mu_{r_0} + \mu_{r_1}) b}{c} \text{, with} \\
  a =~ & \mu_{p_2}^2 \mu_{p_0} \mu_{q_1} \mu_{q_2} \mu_{q_3} + \mu_{p_2}
  (\mu_{p_0} \mu_{q_1} + \mu_{p_0} \mu_{q_2} + \mu_{p_0} \mu_{q_3} + \mu_{q_1} \mu_{q_2} + \mu_{q_1} \mu_{q_3} + \mu_{q_2} \mu_{q_3}) + 1 \text{,} \\
  b =~ & \mu_{p_2} (\mu_{p_0} \mu_{q_1} \mu_{q_2} + \mu_{p_0} \mu_{q_1} \mu_{q_3} + \mu_{p_0} \mu_{q_2} \mu_{q_3} + \mu_{p_1} \mu_{q_2} \mu_{q_3}) +
  \mu_{p_0} + \mu_{q_1} + \mu_{q_2} + \mu_{q_3} \text{ and} \\
  c =~ & (1-\mu_{p_0}^2 \mu_{p_2}) (1-\mu_{p_1}) (1-\mu_{q_1}^2 \mu_{p_2}) (1-\mu_{q_2}^2 \mu_{p_2}) (1-\mu_{q_3}^2 \mu_{p_2}).
  \end{align*}
  Recall Example~\ref{ex:ergodic}, which states that $\|\mathbf{v}\|$
  is finite if and only if all of $1 - \mu_{p_1}$, $1 - \mu_{p_0}^2
  \mu_{p_2}$, $1 - \mu_{q_1}^2 \mu_{p_2}$, $1 - \mu_{q_2}^2 \mu_{p_2}$
  and $1 - \mu_{q_3}^2 \mu_{p_2}$ are positive.  We observe that,
  as suggested above, the denominator $c$ is precisely the product of
  these five factors.
\end{example}

Our approach for computing $\|\mathbf{v}\|$ will involve
first computing variants of $\|\mathbf{v}\|$.  More precisely, for all
subsets $Z$ of $P$, we consider a congruence relation $\sim_Z$ on
markings, such that $m \sim_Z m'$ iff $m$ and $m'$ coincide on all
places $p \in Z$.  Then, we denote by $\mathfrak{M}_Z$ the quotient
set $\mathbb{N}^P / \sim_Z$ and, for every element $\mathfrak{m}$ of
$\mathfrak{M}_Z$, we denote by $\hat{v}(\mathfrak{m})$ the product
$\prod_{p \in Z} \mu_p^{\mathfrak{m}(p)}$.  Two remarkable subsets $Z$
are the set $X = \bigcup_{i \leqslant N-2} P_i \cup
P_{N-1}^{\neg\max}$ and its complement $Y = P_{N-1}^{\max} \cup P_N
= P \setminus X$, which was already mentioned in
Section~\ref{section:ergodic}.  Indeed, places in $X$ are necessarily
bounded while, if the net is unbounded, then so are the places in $Y$.

Based on these objects, and for all integers $c \geqslant 0$, we define
two sets $\mathcal{C}_{m_0}(c)$ and $\mathcal{D}_{m_0}(c)$, which are respective subsets of
$\mathfrak{M}_X$ and of $\mathfrak{M}_Y$. In some sense, the sets
$\mathcal{C}_{m_0}(c)$ are meant to describe the ``bounded'' part of the markings in $\mathcal{R}(m_0)$,
whose ``unbounded'' part is described by the sets $\mathcal{D}_{m_0}(c)$.

\begin{definition}
The set $\mathcal{C}_{m_0}(c)$ is the set of those classes $\mathfrak{m}_X$ in $\mathfrak{M}_X$
such that $\mathfrak{m}_X \cdot P_{N-1}^{\neg\max} = c$, and such that
$\mathfrak{m}_X$ contains some marking in $\mathcal{R}(m_0)$.
The set $\mathcal{D}_{m_0}(c)$ is the set of those classes $\mathfrak{m}_Y$ in $\mathfrak{M}_Y$ such that
$c + \mathfrak{m}_Y \cdot P_{N-1}^{\max} + \mathfrak{m}_Y \cdot \mathsf{cin}(P_N) = \mathbf{C}_{N-1}^{m_0}$, and such that
$\mathfrak{m}_Y$ contains some marking in $\mathcal{R}(m_0)$.
\end{definition}

These two sets of classes allow us to split nicely the huge sum $\|\mathbf{v}\|$ into
smaller independent sums, as stated in the result below.
This decomposition
result has the flavour of convolution algorithms, yet it requires a specific
treatment since its terms are infinite sums.

\begin{restatable}{lemma}{productlem}
\label{productlem}
\label{lemma:product}
The normalisation constant $\|\mathbf{v}\|$ is equal to the following finite sum:
\[\|\mathbf{v}\| = \sum_{c=0}^{|P| \mathbf{W} \|m_0\|} \left(\sum_{\mathfrak{m}_X \in \mathcal{C}_{m_0}(c)} \hat{v}(\mathfrak{m}_X)\right)
\left(\sum_{\mathfrak{m}_Y \in \mathcal{D}_{m_0}(c)} \hat{v}(\mathfrak{m}_Y)\right).\]
\end{restatable}

It remains to compute sums $\sum_{\mathfrak{m}_X \in \mathcal{C}_{m_0}(c)} \hat{v}(\mathfrak{m}_X)$
and $\sum_{\mathfrak{m}_Y \in \mathcal{D}_{m_0}(c)} \hat{v}(\mathfrak{m}_Y)$ for polynomially many values of $c$.
It turns out that such computations can be performed in polynomial time,
as mentioned in Propositions~\ref{proposition:compute-c} and~\ref{proposition:compute-d}.

\begin{restatable}{proposition}{procomputec}
\label{procomputec}
\label{proposition:compute-c}
Let $(\mathcal{N},\lambda,m_0)$ be a live ergodic stochastic open $\mathsf{\Pi}^3$-net,
and let $c \geqslant 0$ be an integer.
There exists an algorithm for computing the sum $\sum_{\mathfrak{m}_X \in \mathcal{C}_{m_0}(c)} \hat{v}(\mathfrak{m}_X)$
in polynomial time with respect to $|P|$, $|T|$, $\mathbf{W}$, $c$ and $\|m_0\|$.
\end{restatable}

Proposition~\ref{proposition:compute-c} is similar to the results
of~\cite{HMN-fi13}, and the algorithm can be adapted to closed
  $\mathsf{\Pi}^3$-nets. The only major improvement here is that, instead of
obtaining an algorithm that is polynomial-time for fixed values of $N$
only, our choice of invariants leads to a polynomial-time algorithm
independently of $N$.

\begin{restatable}{proposition}{procomputed}
\label{procomputed}
\label{proposition:compute-d}
Let $(\mathcal{N},\lambda,m_0)$ be a live ergodic stochastic open $\mathsf{\Pi}^3$-net,
and let $c \geqslant 0$ be an integer.
There exists an algorithm for computing the sum $\sum_{\mathfrak{m}_Y \in \mathcal{D}_{m_0}(c)} \hat{v}(\mathfrak{m}_Y)$
in polynomial time with respect to $|P|$, $|T|$, $\mathbf{W}$, $c$ and $\|m_0\|$.
\end{restatable}

\begin{proof}[Proof (sketch)]
We compute the sum $\sum_{\mathfrak{m}_Y \in \mathcal{D}_{m_0}(c)} \hat{v}(\mathfrak{m}_Y)$
by using a dynamic-programming approach.
Let $a = |P_{N-1}^{\max}|$, $b = |P_N|$, and define $t$ and $u$ to be integers such that
(i) $\mathsf{cin}(p_i^N) > 0$ iff $1 \leqslant i \leqslant t$,
(ii) $\mathsf{cin}(p_i^N) = 0$ iff $t < i < u$ and
(iii) $\mathsf{cin}(p_i^N) < 0$ iff $u \leqslant i \leqslant b$.
In addition, for all $i \leqslant b$, let $\Delta_i = \max\{1,|\mathsf{cin}(p_i^N)|\}$.
We consider below the lattice $\mathbb{L} = \mathbb{Z}^{a-1} \times \prod_{i=1}^b (\Delta_i \mathbb{Z})$.

Now, consider integers $A, B \in \mathbb{Z}$ and $1 \leqslant \alpha \leqslant a$, $1 \leqslant \gamma \leqslant \beta \leqslant b$,
as well as a vector $\mathbf{w}$ in the quotient set $\mathbb{Z}^{a+b-1} / \mathbb{L}$.
We define auxiliary sets of vectors of the form
\[\overline{\mathcal{D}}(A,B,\mathbf{w},\alpha,\beta,\gamma) = \left\{ \mathbf{xy} \in \mathbb{N}^{a+b-1} \left|
\begin{array}{l}A + \sum_{j=u}^{\beta} y_j \geqslant \sum_{i=1}^{\alpha} x_i + \sum_{j=\gamma}^t y_j \\
B  + \sum_{j=u}^{\beta} y_j \geqslant \sum_{j=\gamma}^t y_j \\
\mathbf{xy} - \mathbf{w} \in \mathbb{L} \\
x_i = 0 \text{ for all } i > \alpha \\
y_j = 0 \text{ for all } j < \gamma \text{ and all } j > \beta
\end{array}\right\},\right.\]
where $\mathbf{xy}$ denotes the vector $(x_1,\ldots,x_{a-1},y_1,\ldots,y_b)$.

First, we exhibit a bijection $\mathfrak{m} \mapsto \overline{\mathfrak{m}}$,
from the set $\mathcal{D}_{m_0}(c)$ to a finite union of sets
$\overline{\mathcal{D}}(\mathbf{A},\mathbf{B},\mathbf{0},a-1,b,\mathbf{s})$, where
the integers $\mathbf{A}$, $\mathbf{B}$ and $\mathbf{s}$ can be computed efficiently.
We can also prove a relation of the form $\hat{v}(\mathfrak{m}) = \hat{\vphantom{\rule{1pt}{5.5pt}}\smash{\hat{v}}}(\overline{\mathfrak{m}})$ for all markings
$\mathfrak{m} \in \mathcal{D}_{m_0}(c)$, where $\hat{\vphantom{\rule{1pt}{5.5pt}}\smash{\hat{v}}}$ is a product-form function
$\hat{\vphantom{\rule{1pt}{5.5pt}}\smash{\hat{v}}} : \mathbf{xy} \mapsto \mathbf{P} \prod_{i=1}^{a-1} \nu_{x,i}^{x_i} \prod_{i=1}^{b} \nu_{y,i}^{y_i}$
such that the constants $\mathbf{P}$, $\nu_{x,i}$ and $\nu_{y,j}$ can be computed efficiently.
It remains to compute sums of the form
\[\mathcal{U}(A,B,\mathbf{w},\alpha,\beta,\gamma) = \sum_{\mathbf{xy} \in \overline{\mathcal{D}}(A,B,\mathbf{w},\alpha,\beta,\gamma)} \hat{\vphantom{\rule{1pt}{5.5pt}}\smash{\hat{v}}}(\mathbf{xy}).\]

We do it by using recursive decompositions of the sets $\overline{\mathcal{D}}(A,B,\mathbf{w},\alpha,\beta,\gamma)$,
where the integers $A$ and $B$ belong to the finite set $\{-|\mathbf{A}|-|\mathbf{B}|,\ldots,|\mathbf{A}|+|\mathbf{B}|\}$,
and where the vector $\mathbf{w}$ is constrained to have at most one non-zero coordinate, chosen from some finite domain.
These decompositions involve computing polynomially many auxiliary sums, which will prove Proposition~\ref{proposition:compute-d}.
\end{proof}

We illustrate the last part of this sketch of proof in the case where
$\gamma \leqslant t$, $u \leqslant \beta$, and $\mathbf{w}$ is of the form $\lambda ~ \mathbf{1}_{y,\beta}$,
for some $\lambda \in \mathbb{Z} / \Delta_\beta \mathbb{Z}$ (where we denote by $\mathbf{1}_{y,j}$
the vector of the canonical basis whose unique non-zero entry is $y_j$).
In this case, we show how the sum
$\mathcal{U}(A,B,\mathbf{w},\alpha,\beta,\gamma)$ can be expressed in terms of ``smaller'' sums
$\mathcal{U}(A',B',\mathbf{w}',\alpha',\beta',\gamma')$.
\label{explicit}

Let us split the set $\overline{\mathcal{D}}(A,B,\mathbf{w},\alpha,\beta,\gamma)$ into two subsets:
\begin{align*}
& \overline{\mathcal{D}}_\oplus(A,B,\mathbf{w},\alpha,\beta,\gamma) = \{\mathbf{xy} \in \overline{\mathcal{D}}(A,B,\mathbf{w},\alpha,\beta,\gamma) \mid y_\beta \leqslant y_\gamma\} \text{ and} \\
& \overline{\mathcal{D}}_\ominus(A,B,\mathbf{w},\alpha,\beta,\gamma) = \{\mathbf{xy} \in \overline{\mathcal{D}}(A,B,\mathbf{w},\alpha,\beta,\gamma) \mid y_\beta \geqslant y_\gamma\} \text{, whose intersection is} \\
& \overline{\mathcal{D}}_\odot(A,B,\mathbf{w},\alpha,\beta,\gamma) = \{\mathbf{xy} \in \overline{\mathcal{D}}(A,B,\mathbf{w},\alpha,\beta,\gamma) \mid y_\beta = y_\gamma\}.
\end{align*}
Computing $\mathcal{U}(A,B,\mathbf{w},\alpha,\beta,\gamma)$ amounts to computing the three sums
$\sum_{\mathbf{xy}} \hat{\vphantom{\rule{1pt}{5.5pt}}\smash{\hat{v}}}(\mathbf{xy})$, where $\mathbf{xy}$ ranges respectively on
$\overline{\mathcal{D}}_\oplus(A,B,\mathbf{w},\alpha,\beta,\gamma)$, $\overline{\mathcal{D}}_\ominus(A,B,\mathbf{w},\alpha,\beta,\gamma)$ and
$\overline{\mathcal{D}}_\odot(A,B,\mathbf{w},\alpha,\beta,\gamma)$.
We show here how to compute the first sum, which we simply denote by $\mathcal{U}_\oplus$.
The two latter sums are computed similarly.

Splitting every vector $\mathbf{xy}$
into one vector $y_\beta(\mathbf{1}_{y,\gamma}+\mathbf{1}_{y,\beta})$ and one vector
$\mathbf{xy}' = \mathbf{xy} - y_\beta(\mathbf{1}_{y,\gamma}+\mathbf{1}_{y,\beta})$,
we observe that $\mathbf{xy} - \lambda ~ \mathbf{1}_{y,\beta} \in \mathbb{L}$ if and only if
(i) $y_\beta \equiv \lambda \mod{\Delta_\beta}$, and
(ii) $\mathbf{xy}' - y_\beta \mathbf{1}_{y,\gamma} \in \mathbb{L}$.
It follows that
\begin{align*}
\hspace*{-5mm}\overline{\mathcal{D}}_\oplus(A,B,\lambda ~ \mathbf{1}_{y,\beta},\alpha,\beta,\gamma) =~ &
\bigsqcup_{j \geqslant 0} \{\mathbf{xy} \in \overline{\mathcal{D}}_\oplus(A,B,\lambda ~ \mathbf{1}_{y,\beta},\alpha,\beta,\gamma) \mid y_\beta = j\} \\
=~ & \bigsqcup_{j \geqslant 0, j \equiv \lambda \!\!\!\mod{\Delta_\beta}}
\left(\overline{\mathcal{D}}(A,B,j ~ \mathbf{1}_{y,\gamma},\alpha,\beta-1,\gamma) + j (\mathbf{1}_{y,\gamma}+\mathbf{1}_{y,\beta}) \right).
\end{align*}

Then, using the change of variable $j = k + \Delta_\gamma \Delta_\beta \ell$ (with $0 \leqslant k < \Delta_\gamma \Delta_\beta$),
observe that $j \equiv k \!\!\!\mod{\Delta_\beta}$ and that
$j ~ \mathbf{1}_{y,\gamma} \equiv k ~ \mathbf{1}_{y,\gamma} \mod \mathbb{L}$, whence
\begin{align*}
\mathcal{U}_\oplus =~ &
\sum_{j \geqslant 0} \mathbf{1}_{j \equiv \lambda \!\!\!\mod{\Delta_\beta}} ~
\nu_{y,\gamma}^j ~ \nu_{y,\beta}^j ~ \mathcal{U}(A,B,j ~ \mathbf{1}_{y,\gamma},\alpha,\beta-1,\gamma) \\[-.2cm]
=~ & \sum_{\ell \geqslant 0} \sum_{k=0}^{\Delta_\gamma \Delta_\beta - 1} 
\mathbf{1}_{k \equiv \lambda \!\!\!\mod{\Delta_\beta}} ~
(\nu_{y,\gamma} \nu_{y,\beta})^{k + \Delta_\gamma \Delta_\beta \ell} \mathcal{U}(A,B,k ~ \mathbf{1}_{y,\gamma},\alpha,\beta-1,\gamma) \\[-.2cm]
=~ & \frac{1}{1-(\nu_{y,\gamma} \nu_{y,\beta})^{\Delta_\gamma \Delta_\beta}} \sum_{k=0}^{\Delta_\gamma \Delta_\beta - 1}
\mathbf{1}_{k \equiv \lambda \!\!\!\mod{\Delta_\beta}} ~
(\nu_{y,\gamma} \nu_{y,\beta})^k \mathcal{U}(A,B,k ~ \mathbf{1}_{y,\gamma},\alpha,\beta-1,\gamma).
\end{align*}

This shows, in this specific case, that computing $\mathcal{U}_\oplus$ reduces to computing
finitely many sums of the form $\mathcal{U}(A,B, \mathbf{w}',\alpha,\beta-1,\gamma)$.
Similar constructions are successfully used in all other cases.

\section{Conclusion}

  Performance analysis of infinite-state stochastic systems is a
  very difficult task. Already checking the ergodicity is difficult in
  general, and even for systems which are known to be ergodic and
  which have product-form steady-state distributions, computing the
  normalising constant can be hard. In this work, we have proposed the
  model of open $\mathsf{\Pi}^3$-nets; this model generalises queuing
  networks and closed produc-form Petri nets
  and generates a potentially infinite state-space. We have
  shown that we can efficiently decide (in polynomial time!) many
  behavioural properties, like the boundedness, the reachability in live nets, and
  the most important quantitative property: ergodicity. Furthermore, using dynamic
  programming algorithms managing infinite sums, we have shown that we can compute the
  normalising constant of the steady-state distribution in
  pseudo-polynomial time.

  We believe our approach can be extended to $\mathsf{\Pi}^3$-nets in which
  one place is removed in every layer, without affecting too much the
  complexity. This setting would allow to model production
 of  resources by the environment while in the current version resources
 may grow in an unbounded way but only when the number of processes
 of the main layer also grows.
We leave this as future work.

\newpage
\appendix

\bigskip
\bigskip
\begin{center}
\Huge\textbf{--- Appendix ---}
\end{center}
\bigskip

\section{Proofs of Section~\ref{sec:formalism}}

\proppideuxtrois*

\begin{proof}
We first deal with the case where $\net$ is closed.
By definition of a closed \pitrois-net, every connected
component of $G_{\net}$ is strongly connected.
Hence, it remains to prove that every bag $b$ of $V_{\net}$
admits a witness (denoted $\wit_b$).

For every $i \leqslant N$ and every place $p \in P_i$,
we construct the witness of the bag $b_p$ by backward
induction on $i$, as follows:
\[\wit_{b_p} = \begin{cases} p & \text{if } i = N \text{ and } p \neq p_\ext; \\
p_\ext - P_N & \text{if } p = p_\ext; \\
p - \sum_{\beta \in V_{i+1}} \beta(p) \wit_\beta & \text{if } i \leqslant N-1.\end{cases}\]
Let us prove that these vectors $\wit_{b_p}$ are indeed witnesses of the bags $b_p$.

First, if for all places $p \neq p_\ext$,
the support of $\wit_{b_p}$ is a subset of $\bigcup_{j \geqslant i} P_j \setminus \{p_\ext\}$.
It comes at once that, for all transitions $t \in \bigcup_{j \leqslant i} T_j$, we have
\[\begin{cases}
\wit_{b_p}\cdot W(t) =-1 & \mbox{if } W^{-}(t) = b_p;\\
\wit_{b_p}\cdot W(t) =1 & \mbox{if } W^{+}(t)= b_p;\\
\wit_{b_p}\cdot W(t) =0 & \mbox{otherwise}.
\end{cases}\]

This already proves that $\wit_{b_p}$ is a witness of $b_p$ whenever $p \in P_N \setminus \{p_\ext\}$.
Let us also prove that $\wit_{b_{p_\ext}}$ is a witness of $b_{p_\ext}$.
We observe that $\sum_{p \in P_N} \wit_{b_p} = P_N - P_N = 0$.
Hence, for every transition $t \notin T_N$, then $\wit_{b_{p_\ext}}\cdot W(t) = 0$.
Moreover, consider some transition $t \in T_N$, and let $b_q = W^{-}(t)$ and $b_r = W^+(t)$:
\begin{itemize}
 \item if $q = p_\ext$, then $r \neq p_\ext$, hence $\wit_{b_{p_\ext}}\cdot W(t) = -\wit_{b_r}\cdot W(t) = 1$;
 \item if $r = p_\ext$, then $q \neq p_\ext$, hence $\wit_{b_{p_\ext}}\cdot W(t) = -\wit_{b_q}\cdot W(t) = 1$;
 \item if $q$ and $r$ are both distinct from $p_\ext$, then $\wit_{b_{p_\ext}}\cdot W(t) = -\wit_{b_q} \cdot W(t) - \wit_{b_r} \cdot W(t) = 0$.
\end{itemize}
This proves that $\wit_{b_{p_\ext}}$ is a witness of $b_{p_\ext}$.

Second, by induction hypothesis, and for all transitions $t \in \bigcup_{j > i+1} T_j$
and all bags $\beta \in V_{i+1}$, we cannot have $\beta = W^-(t)$ or $\beta = W^+(t)$,
hence we have $\wit_\beta \cdot W(t) = 0$. Since $p$ does not belong to the support
of $W^-(t)$ nor of $W^+(t)$, it also follows that $p \cdot W(t) = 0$, whence $\wit_{b_p} \cdot W(t) = 0$.

Finally, consider some transition $t \in T_{i+1}$, and let $\beta = W^-(t)$ and $\beta' = W^+(t)$.
Using again the induction hypothesis, we have
\[\wit_{b_p} = (p - \beta(p) \wit_\beta - \beta'(p) \wit_{\beta'}) \cdot W(t) = (p \cdot W(t)) + \beta(p) - \beta'(p) = 0,\]
which proves that $\wit_{b_p}$ is indeed a witness of $b_p$ and
completes the proof in the case of closed nets.

\medskip
We deal now with the case where $\net$ is an open net,
based on the closed net $\overline{\net}$.
In this case, every bag of $\overline{\net}$ remains
a bag of $\net$, except the bag $b_{p_\ext}$,
which is transformed into a new bag $b_{p_\ext}^\star := b_{p_\ext} - p_\ext$.
We also write $b = b^\star$ for all other bags of $\overline{\net}$.
Accordingly, we denote by $\wit_b$ the witness of
a bag $b$ in $\overline{\calN}$.

Recall that $p_\ext$ does not belong to the support of $\wit_b$,
so that, for every transition of $\calN$ or $\overline{\calN}$,
the quantity $\wit_b \cdot W(t)$ does not depend on whether
the place $p_\ext$ was deleted.
Then, we consider two cases:
\begin{enumerate}
\item If $b_{p_\ext}^\star = \beta$ for some bag $\beta$ of $\overline{\net}$,
then the bag graph $\calG_\net$ is obtained from the original bag
graph $\calG_{\overline{\net}}$ by merging the vertices
$b_{p_\ext}$ and $\beta$ (the resulting vertex is named $\beta$).
Since every connected component of $\calG_{\overline{\net}}$
is strongly connected, so is every connected component of $\calG_\net$.

Then, let $b = b^\star$ be a bag of $\net$ distinct from $\beta$.
The transitions of $\calN$ with input (resp. output) bag $b^\star$ are
exactly those transitions of $\overline{\calN}$
with input (resp. output) bag $b$.
It follows that $\wit_b$
is still a witness of $b^\star$ in $\net$.

Moreover, the transitions of $\calN$
with input (resp. output) bag $\beta$ are exactly
those transitions of $\overline{\calN}$
with input (resp. output) bag $b_{p_\ext}$ or $\beta$.
It follows that $\wit_\beta + \wit_{b_{p_\ext}}$
is a witness of $\beta$ in $\net$.
 
\item If $b_{p_\ext}^\star$ is not equal to any bag of $\overline{\net}$,
then the bag graph $\calG_\net$ is obtained from
$\calG_{\overline{\net}}$ by renaming
the vertex $b_{p_\ext}$ into the new vertex $b_{p_\ext}^\star$.
Hence, every connected component of $\calG_\net$ is strongly connected.

Furthermore, for every bag $b^\star$ of $\net$,
the transitions of $\calN$ with input (resp. output) $b^\star$
are exactly those transitions of $\overline{\calN}$
whith input (resp. output) bag $b$
It follows that $\wit_b$
is still a witness of $b^\star$ in $\net$.
\end{enumerate}
This proves that both open and closed \pitrois-nets are \pideux-nets.
\end{proof}

\section{Proofs of Section~\ref{sec:qualitative}}

We also refine the notion of potential that appeared in
Definition~\ref{def:pitrois}. Let $b_p$ be the bag corresponding to
place $p \in P_i$; it can be rewritten as $p+\sum_{q \in P_{i-1}}
b_p(q)\ q$ where $b_p(q) \in \nat$. The value $b_p(q)$ corresponds to
the potential of $p$ w.r.t. $q$; it corresponds to the number of
$q$-resources that are required to fire transitions out of $b_p$. It
is easy to see that $\pot(p) = \sum_{q \in P_{i-1}} b_p(q)$.  The same
can be done for the external place $p_{\ext}$.

First, as an obvious consequence of the definitions, we have:

\begin{lemma}
  \label{lemma:transition}
  Let $m$ be a marking, and assume that $m \xrightarrow{t} m'$ with $t
  \in T_i$ (assuming $2 \le i \le N$). Let $p_1$, resp. $p_2$, be the
  input and output place of $t$ in $P_i$ (note that, by convention,
  one of those can be $p_\ext$). Note that, since we assumed that
  no transition is useless, we cannot have $p_1 = p_2$. Then:
  \begin{itemize}
  \item $m'(p_1) = m(p_1) -1$ and $m'(p_2) = m(p_2) +1$;
  \item for every place $p \in P_i \setminus \{p_1,p_2\}$, $m'(p) = m(p)$
  \item for every place $q \in P_{i-1}^{\max}$,
  $m'(q) = m(q) + \big(b_{p_2}(q) - b_{p_1}(q)\big)$
  \item for every place $r \in P \setminus (P_{i-1}^{\max} \cup P_i)$,
  $m'(r) = m(r)$.
  \end{itemize}
\end{lemma}





\subsection{Liveness analysis}
\label{app:liveness}

For proving Theorem~\ref{theo:liveness}, we establish some preliminary
technical results.

\begin{lemma}
  \label{lemma:minimal-marking}
  Fix $1 \le i \le N$.  Let $m$ be a marking, and $m' \in
  \calR_{\net}(m)$.  Then:
  \[
  m \in \Live_i\ \text{implies}\ m' \in \Live_i
  \]
\end{lemma}

\begin{proof}
  Pick a transition $t \in T_i$ with $1 \le i \le N$, and assume $m
  \xrightarrow{t} m'$.  Only markings of places in $P_i$ and $P_{i-1}^{\max}$
  are affected by the firing of this transition.
  Moreover, membership in the set $\Live_j$ depends only on
  the number of tokens gathered in $P_j \cup P_{j+1}^{\neg \max}$.
  Hence only constraints $\Live_i$ and $\Live_{i-1}$ may be
  affected.

  Let $p_1$, resp. $p_2$, be the input, resp. output, place of $t$ in
  $P_i^\star$ (as previously, $p_1$ or $p_2$ can be $p_{\ext}$).  We
  analyze the three cases separately:
  \begin{itemize}
  \item Assume $i=N$. Then the property obviously holds for $\Live_N$.
  \item Assume $1 \le i \le N-1$. We notice that $m' \cdot P_i =
    m\cdot P_i$, and that the right-hand side of the condition
    defining $\Live_i$ is unchanged. So $m \in \Live_i$ implies $m'
    \in \Live_i$.
  \item Assume $1 \le i-1 \le N-1$. We have that
    \[
    m' \cdot P_{i-1} = m \cdot P_{i-1} -\pot(p_1) + \pot(p_2)\
    \text{and}\ m \cdot P_{i-1} \ge \pot(p_1).
    \]
    So we deduce that $m' \cdot P_{i-1} \ge \pot(p_2)$. Since either
    $m'(p_2)>0$ or $p_2 = p_{\ext}$, we deduce that $m \in
    \Live_{i-1}$ implies $m' \in \Live_{i-1}$.
  \end{itemize}
\end{proof}

We now recall the notions of $i$-liveness of~\cite{HMN-fi13}.  Let $m$
be a marking. For convenience, when $\bowtie$ is a comparison operator
and $1 \le i \le N$, we define $\calR_{\bowtie i}(m)$ the set of all
markings reachable from $m$ by firing sequences of transitions taken
in $T_{\bowtie i} \stackrel{\text{def}}{=} \bigcup_{j \bowtie i}
T_j$. We say that $m$ is \emph{$i$-live} if for every transition $t
\in T_{\le i}$, there is a marking in $\calR_{\le i}(m)$ which enables
$t$.

\begin{lemma}
  \label{lemma:+1}
  Fix $1 \le i \le N$, and pick $p \in P_i$ such that for every $q \in
  P_{i-1}$, $m(q) \ge b_p(q)$. For every $p' \in P_i$ (with $p' \ne
  p$), there is $m' \in \calR_{=i}(m)$ such that $m'(p)=m(p)-1$,
  $m'(p')=m(p')+1$, $m'(-)=m(-)$.\footnote{This notation is for
    ``$\forall p'' \in P_i \setminus \{p,p'\}$, $m'(p'') = m(p'')$''.}

  Pick $p \in P_N$ such that for every $q \in P_{N-1}$, $m(q) \ge
  b_p(q)$. There is $m' \in \calR_{=N}(m)$ such that $m'(p)=m(p)-1$,
  $m'(-)=m(-)$.

  Assume that for every $q \in P_{N-1}$, $m(q) \ge b_{p_\ext}(q)$. For
  every $p' \in P_N$, there is $m' \in \calR_{=N}(m)$ such that
  $m'(p') = m(p')+1$ and $m'(-) = m(-)$.
\end{lemma}

\begin{proof}
  This property is rather obvious: any transition $t$ out of $p$ is
  enabled, and if $m \xrightarrow{t} m'$, then $m'(q) \ge b_p(q)$ for
  every $q \in P_{i-1}$ as well. We can repeat iteratively, and fire
  any sequence of transitions in $T_i$ leading from $p$ to $p'$.

  The second and third properties can be treated similarly.
\end{proof}

\begin{lemma}
  \label{lemma:enable}
  Fix $1 \le i \le N$. Let $p \in P_i^\star$.  Assume $m \cdot P_{i-1}
  \ge \pot(p)$ and $m \in \bigcap_{j<i} \Live_j$. Then there is $m'
  \in \calR_{\le i-1}(m)$ such that for every $q \in P_{i-1}$, $m'(q)
  \ge b_p(q)$
\end{lemma}

\begin{proof}
  We do the proof by induction on $i$. The case $i=1$ is obvious since
  layer $1$ of $\calN$ is a state machine.  Assume that $i>1$ and the
  result holds for $i-1$.

  We define $\nu_p(m) = \sum_{q \in P_{i-1}}
  \max(0,(b_p(q)-m(q)))$. This measures how many tokens are missing
  for enabling transitions of $T_i$ out of $p$.  In particular,
  $\nu_p(m) \ge 0$, and transitions out of $p$ are enabled by $m$ if
  and only if $\nu_p(m)=0$.  So, if $\nu_p(m)=0$, there is nothing to
  be done, and we can choose $m' = m$. We therefore assume that
  $\nu_p(m)>0$, and we show the following intermediary result:
  
  \begin{lemma}
    There is some $q \in P_{i-1}$ such that $m(q)>b_p(q)$ and
    $m \cdot P_{i-2} \ge \pot(q)$.
  \end{lemma}
  \begin{proof}
    Since $m\cdot P_{i-1} \ge \pot(p) = \sum_{q \in P_{i-1}} b_p(q)$
    and $\nu_p(m)>0$, there exists some place $q \in P_{i-1}$ such
    that $m(q)-b_p(q)>0$ (there are tokens which are not ``reserved''
    for firing a transition of $T_i$ out of $p$).

    If $m\cdot P_{i-2} \ge \POT_{i-1}$, then any choice of $q$ above
    will satisfy the expected property.

    Assume that $m \cdot P_{i-2} < \POT_{i-1}$. Since $m \in
    \Live_{i-2}$, this means that there is some $q \in P_{i-1}$ such
    that $m(q)>0$ and $\pot(q) \le m \cdot P_{i-2} <\POT_{i-1}$: since
    $\calN$ is a \pitrois-net, $q$ cannot be an interface place! In
    particular, $b_p(q)=0$. This place $q$ satisfies the expected
    property.
  \end{proof}

  We apply the induction hypothesis: there is $m_1 \in \calR_{\le
    i-2}(m)$ such that $m_1$ enables the transitions out of $q$ (that
  is, for every $r \in P_{i-2}$, $m_1(r) \ge b_q(r)$).  Fix now some
  $q' \in P_{i-1}$ such that $m(q')<b_p(q')$. Applying
  Lemma~\ref{lemma:+1}, there exists $m_2 \in \calR_{=i-1}(m_1)
  \subseteq \calR_{\le i-1}(m)$ such that $m_2(q)=m_1(q)-1=m(q)-1$,
  $m_2(q')=m_1(q')+1=m(q')+1$, $m_2(-) = m_1(-) = m(-)$.

  We get in particular that $\nu_p(m_2) = \nu_p(m)-1$. Furthermore,
  $m_2 \cdot P_{i-1} = m \cdot P_{i-1} \ge \pot(p)$, and thanks to
  Lemma~\ref{lemma:minimal-marking}, for every $j<i$, $m_2 \in
  \Live_j$.  We can therefore iterate the process and build a (finite)
  sequence of markings reachable from $m$, whose $\nu_p$-value
  decreases until reaching $0$. The last marking is the expected one.
\end{proof}

Finally, we show the following result, which will directly imply
Theorem~\ref{theo:liveness}.
\begin{lemma}
  \label{lemma:i-liveness}
  Fix $1 \le i \le N$. A marking $m$ is $i$-live if and only if $m \in
  \bigcap_{j=1}^{i-1} \Live_j$ and $m \cdot P_i > 0$ unless $i=N$ and
  $\calN$ is open.
\end{lemma}

\begin{proof}
  This precise result was proven in~\cite{HMN-fi13} in the case of
  closed nets. So, the result holds directly if $\calN$ is closed or
  if $i<N$. So we assume $i=N$ and the net is open.

  We define the marked net $(\overline{\calN},\overline{m})$ from the
  marked net $(\calN,m$) by adding the place $p_\ext$, and by adding a
  token into the place $p_\ext$.  The marked net
  $(\overline{\calN},\overline{m})$ has fewer behaviours than the
  original net $(\calN,m)$: every sequence of transitions that can be
  fired from $\overline{m}$ in $\overline{\calN}$ can also be fired
  from $m$ in $\calN$.  .  Moreover, if $(\calN,m)$ respects the
  conditions of Lemma~\ref{lemma:i-liveness} for being $N$-live (open
  net), then $(\overline{\calN},\overline{m})$ also satisfies the
  conditions the conditions of Lemma~\ref{lemma:i-liveness} for being
  $N$-live (closed net). Applying the result of~\cite{HMN-fi13} for
  closed nets, we get that $(\overline{\calN},\overline{m})$ is
  live. Every transition can eventually be fired in $\overline{\calN}$
  from $\overline{m}$, hence it is also the case in $\calN$ from
  $m$. By weak-reversibility of the net $\calN$, this is also the case
  from every marking $m'$ reachable from $m$. Hence, $(\calN,m)$ is
  $N$-live too.

  Conversely, if $(\calN,m)$ does not respect these conditions, it
  means that $m \notin \Live_j$ for some $j \le N-1$.  By
  weak-reversibility of the net, we know that, for all $m' \in
  \calR_\calN(m)$, we have $m \in \calR_\calN(m')$. Using
  Lemma~\ref{lemma:enable}, it follows that $m' \notin \Live_j$.
  In particular, no transition in $T_{j+1}$ is enabled by $m'$, which
  proves that $m$ is not $N$-live.

\end{proof}

\subsection{Reachable markings}

\subsubsection{Linear invariants}

\propinvariants*

\label{app:invariants}

\begin{proof}
  Pick $t \in T_i$ for some $1 \le i \le N$. We write $p_1$
  (resp. $p_2$) for the input (resp. output) place of $t$ in $P_i$.

  We consider one of the vectors $v_j$, with $1 \le j \le N$ (or
  $N-1$). Since the support of $v_j$ is included in $P_j \cup
  P_{j+1}$, if $j<i-2$ or $j>i$, then obviously, $v_j \cdot W(t)=0$.

  We now consider $v_{i-2}$ (assuming $1 \le i-2$). We compute:
  \[
  v_{i-2} \cdot W(t)  = \sum_{q \in P_{i-1}} \cin(q)\ 
  (b_{p_2}(q)-b_{p_1}(q))
  \]
  Now, if $b_{p_2}(q)>0$ or $b_{p_1}(q)>0$, then this means that $q$
  is an interface place; hence, $q$ has maximal potential in slice
  $i-1$ ($\calN$ is a $\Pi_3$-net); this implies that $\pot(q) =
  \POT_{i-1}$, hence $\cin(q)=0$.  We conclude that $v_{i-2} \cdot
  W(t) = 0$.

  We now consider $v_{i-1}$ (assuming $1 \le i-1$). We compute:
  \begin{multline*}
    v_{i-1} \cdot W(t) = -\cin(p_1)+\cin(p_2) +\sum_{q \in P_{i-1}}
    (b_{p_2}(q) - b_{p_1}(q)) \\
    = -\cin(p_1)+\cin(p_2) + \pot(p_2)-\pot(p_1) = 0
  \end{multline*}

  We consider $v_i$ when $i \le N-1$. We immediately get $v_i \cdot
  W(t) = -1+1=0$.

  We finally consider $v_i$ when $i=N$ and the net is closed. We
  immediately get as well $v_N \cdot W(t) = -1+1 =0$.
\end{proof}

\subsubsection{Characterization of the reachability set of live nets}
\label{app:reach}

Theorem~\ref{theo:reach} will be a consequence of the following lemma:

\begin{lemma}
  \label{lem:CNS-reach}
  \label{lem:reach}
  Let $i \leqslant N$, and suppose that the initial marking $m_0$ is
  $i$-live.  Then:
  \[
  \calR_{\leqslant i}(m_0) = \bigcap_{j=1}^i \Inv_j(m_0) ~\cap
  ~\bigcap_{j=1}^i \Live_j ~\cap ~ \{m \mid \forall p \in
  \bigcup_{j=i+1}^N P_j,\ m(p) = m_0(p)\}
  \]
\end{lemma}

\begin{proof} 
  We write $\calS_{\leq i}(m_0)$ for the right hand-side of the
  equality in the statement.  The inclusion $\calR_{\leq i}(m_0)
  \subseteq \calS_{\leq i}(m_0)$ is a consequence of
  Corollary~\ref{coro:invariants} and Theorem~\ref{theo:liveness}.
  
    Conversely, let $p_1,\ldots,p_N$ be places in $P_1^{\max}, \ldots,
    P_N^{\max}$.  For every marking $m$, we denote by $\pi_i(m)$ the
    marking $m'$ such that: (i) $m'(q) = q$ for all places $q \in P_j$
    for some $j > i$; (ii) $m'(p_i) = m \cdot P_i$; (iii) $m'(p_j) =
    \bfC_j^m$ for all $j < i$; (iv) $m'(q) = 0$ for all places $q \in
    P_j \setminus \{p_j\}$ for some $j \leqslant i$.  Let us also
    assume that $\calN$ is closed. 

    We prove by induction on $i$ and $\ell$ the following property,
    denoted by $\bfP(i,\ell)$: if $m \in \calS_{\leq i}(m_0)$ and if
    $m \cdot P_i = m(p_i) + \ell$, then $\pi_i(m) \in \calR_{\leqslant
      i}(m)$.  We also denote by $\bfP(i)$ the conjunction of the
    properties $\bfP(i,\ell)$ for all $\ell \geqslant 0$.
  
    Since the net is weakly reversible, observe that $\calR_{\leqslant
      i}(m') = \calR_{\leqslant i}(m)$ for all markings $m$ and $m'$
    such that $m' \in \calR_{\leqslant i}(m)$.  In addition, observe
    that if $m_0$ is $i$-live and if $m \in \calS_{\leqslant i}(m_0)$,
    then $m$ is $i$-live, $\calS_{\leqslant i}(m) = \calS_{\leqslant
      i}(m_0)$ and $\pi_i(m) = \pi_i(m_0)$.

    Since $\bfP(0)$ is immediate, we assume now that $i \geqslant 1$
    and that $\bfP(i-1)$ holds.  Then, if $m \cdot P_i = m(p_i)$, observe
    that $\pi_i(m) = \pi_{i-1}(m)$, hence $\bfP(i,0)$ follows from
    $\bfP(i-1)$. Therefore, we assume now that $\ell \geqslant 1$ and
    that both $\bfP(i,\ell-1)$ and $\bfP(i-1)$ hold.
  
    In that case, let $p$ be a place in $P_i$ such that $m(p)
    \geqslant 1$ and such that $\pot(p)$ is as small as
    possible. Since $\pi_{i-1}(m) \in \calR_{\le i-1}(m) \subseteq
    \calS_{\le i-1}(m) = \calS_{\le i-1}(m_0)$, we have that
    $\pi_{i-1}(m) \models \Live_{i-1}$; hence $\pi_{i-1}(m) \cdot
    P_{i-1} \ge \pot(p)$ and we can apply Lemma~\ref{lemma:enable} and
    Lemma~\ref{lemma:+1}: there is a marking $m' \in \calR_{\le
      i-1}(\pi_{i-1}(m))$ such that $m'(p_i) =
    \pi_{i-1}(m)(p_i)+1=m(p_i)+1$, $m'(p) = \pi_{i-1}(m)(p)-1=m(p)-1$,
    and $m'(-)=m(-)$ otherwise on $P_i$. We can apply $\bfP(i,\ell-1)$
    to $m'$: we get that $\pi_i(m') \in \calR_{\le i}(m')$. Since $m'
    \in \calR_{\le i}(\pi_{i-1}(m))$, we get by $\bfP(i-1)$ that $m'
    \in \calR_{\le i}(m)$, hence $\calR_{\le i}(m) = \calR_{\le
      i}(m')$. Finally notice that $\pi_i(m)=\pi_i(m')$, hence we
    conclude that $\pi_i(m) \in \calR_{\le i}(m)$, which concludes the
    proof of $\bfP(i,\ell)$.

    \medskip We assume $\calN$ is an open net. Consider some marking
    $m \in \calS^{\calN}_{\leq i}(m_0)$ (the superscript of the
    notation is for distinguishing the net which is considered).  We
    define the closed net $\overline{\calN}$ by adding the place
    $p_\ext$.  We also extend both markings $m_0$ and $m$ to markings
    $\overline{m_0}$ and $\overline{m}$, where $\overline{m_0}(p_\ext)
    = 1 + m \cdot P_N$ and $\overline{m}(p_\ext) = 1 + m_0(P_N)$.  By
    construction, if $i<N-1$, we have $\calS^{\overline{\calN}}_{\leq
      i}(\overline{m_0}) = \calS^{\calN}_{\leq i}(m_0)$ since the
    various constraints do not involve places in layer $N$. Hence
    $\overline{m} \in \calS^{\overline{\calN}}_{\leq
      i}(\overline{m_0}) \subseteq \calR^{\overline{\calN}}_{\leq
      i}(\overline{m_0})$ in the closed net $\overline{\calN}$. The
    closed net has fewer behaviours than the original net $\calN$,
    hence $m \in \calR^\calN_{\leqslant i}(m_0)$, which implies
    $\calS_{\leq i}(m_0) \subseteq \calR_{\leq i}(m_0)$.

    Assume now that $i=N-1$. Since $\overline{m_0}(p_\ext) \geqslant 1$,
    the marking $m_0$ belongs to $\Live_{N-1}(\calN)$ if and only if
    $\overline{m_0}$ belongs to $\Live_{N-1}(\overline{\calN})$.
    Similarly, $m$ belongs to $\Live_{N-1}(\calN)$ if and only if
    $\overline{m}$ belongs to $\Live_{N-1}(\overline{\calN})$.
    Yet, the values of $\cin^\calN$ and
    $\cin^{\overline{\calN}}$ on layer $i+1=N$ might differ, since we have
    $\cin^{\overline{\calN}}(p) = \POT^{\overline{\calN}}_N - \pot(p) =
    \cin^\calN(p) + \POT^{\overline{\calN}}_N -\pot(p_\ext)$ for all $p \in P_N$.
    Hence we first show that $\overline{m} \in \calS^{\overline{\calN}}_{\leq
    i}(\overline{m_0})$:
    \begin{multline*}
      \overline{m} \cdot (P_{N-1} + \cin^{\overline{\calN}}(P_N^\star)) \\
        ~~~~ = m \cdot P_{N-1} + m \cdot \cin^\calN(P_N) +
        (\POT^{\overline{\calN}}_N -\pot(p_\ext)) m \cdot P_N + m(p_\ext) \cin^{\overline{\calN}}(p_\ext) \\
        ~~~~ = \bfC_{N-1}^{m_0}(\calN) + (\POT^{\overline{\calN}}_N -\pot(p_\ext)) (m \cdot P_N - m(p_\ext)) \\
        ~~~~ = \bfC_{N-1}^{m_0}(\calN) - (\POT^{\overline{\calN}}_N -\pot(p_\ext)).\hfill
    \end{multline*}
    Similarly, $\overline{m_0} \cdot (P_{N-1} + \cin^{\overline{\calN}}(P_N^\star))
      = \bfC_{N-1}^{m_0}(\calN) - (\POT^{\overline{\calN}}_N -\pot(p_\ext))$,
    which implies $\overline{m} \in \calS^{\overline{\calN}}_{\le
      i}(\overline{m_0})$. We then apply the same reasoning as in the
    previous case.

    Finally, assume that $i=N$. Obviously,
    $\overline{m} \cdot (P_N^\star) = m \cdot P_N + m_0 \cdot P_N +1 = \overline{m_0} \cdot (P_N^\star)$,
    and the same reasoning applies as well.

    This proves the lemma in the case of open nets too.
\end{proof}

\subsubsection{The boundedness problem of live nets}

\coroboundness*

\label{app:boundedness}

\begin{proof}
  Checking that the marking $m_0$ is live is feasible in polynomial
  time, hence it remains to decide, given a live marking $m_0$,
  whether the net is bounded.  If the net is closed, then
  Theorem~\ref{th:pitrois} proves that the net is bounded. Hence, we
  focus only on open nets: we prove that, in that case, the net is
  bounded if and only if $\cin(q) > 0$ for all $q \in P_N$.

  Indeed, if $\cin(q) > 0$ for all $q \in P_N$, then we consider the
  ``sum'' invariant: $\sum_{p \notin P_N} m(p) + \sum_{p \notin P_1}
  \cin(p) m(p) = \sum_{i=1}^{N-1} \bfC_i^{m_0}$.  Since $\cin$ is
  non-negative on every place outside $P_N$ and positive on $P_N$, it
  follows that $m(p) \leqslant \sum_{i=1}^{N-1} \bfC_i^{m_0}$ for all
  places.

  Conversely, if $\cin(q) \leqslant 0$ for some place $q \in P_N$,
  consider some place $p \in P_{N-1}^{\max}$. Theorem~\ref{theo:reach}
  proves that, for all $n \geqslant 0$, the marking $m_0 + n(q +
  |\cin(q)| p)$ is reachable from $m_0$, whence the net is unbounded.
\end{proof}

  



\subsubsection{The boundedness problem of non-live nets}

\coNPhardness*

\label{app:coNPhardness}

\begin{proof}
  We reduce the independent-set problem to the non-boundedness problem
  of a $\Pi^3$-net. The independent-set problem is defined as follows:
  given an undirected graph $G=(V,E)$ and an integer $k$, does there
  exist an independent set of size $k$ in $G$. A set $X \subseteq V$
  of vertices is \emph{independent} whenever for every $x,y \in X$,
  $\{x,y\} \notin E$. If $v \in V$, we denote $E(v)$ the set of
  vertices adjacent to $v$, that is $\{v' \in V \mid \{v,v'\} \in
  E\}$.

  Fix an instance of this game $G = (V,E)$ and an integer $k$. Write
  $n = |V|$. We build the following 4-open \pitrois-net $\net$:
  \begin{itemize}
  \item $P_1 = \{p_1\}$, $P_2= E \cup \{p_2\}$, $P_3 = V \cup
    \{p_3\}$, $P_4 =\{p_4\}$, where $p_1,p_2,p_3,p_4$ are fresh
    symbols;
  \item We define the transitions of $\net$ through the bag graph
    (since this is more readable):
    \[
    \begin{array}{ll}
      \text{Layer}\ 1: & \text{no transition} \\
      \text{Layer}\ 2: & \text{for every} \ e \in E,\ e+p_1 \leftrightarrow p_2+p_1 \\
      \text{Layer}\ 3: & \text{for every}\  v \in V,\ v+E(v) \leftrightarrow p_3 + n\, p_2 \\
      \text{Layer}\ 4: & k \, p_3 \leftrightarrow p_4 + k\, p_3
    \end{array}
    \]
    where $b \leftrightarrow b'$ ($b$ and $b'$ being bags) is a
    shorthand for $b \to b'$ and $b' \to b$.
%
%
%
  \end{itemize}
  This is easy to check that this is a \pitrois-net (the place $p_3$
  is of maximal potential in layer $3$, whereas all places of layer
  $2$ have maximal potential ($1$ in this case)).

  We let $m_0 = E+V$, and we will show that $(\net,m_0)$ is bounded if
  and only if $V$ contains an independent subset of size $k$.

  Assume that $V$ contains an independent subset $X$ of size
  $k$. Since $X$ is an independent set, for every $x, y \in X$, $E(x)
  \cap E(y) = \emptyset$. We can then apply the transitions $x+E(x)
  \to p_3 + n\, p_2$ for every $x \in X$. The resulting marking is
  $k\, p_3 + nk\, p_2$. The transition $k \, p_3 \to p_4 + k\, p_3$ is
  therefore enabled, and can be taken an arbitrary number of times,
  producing markings with an arbitrary number of tokens in place
  $p_4$. We deduce that $\net$ is not bounded.

  We assume that $\net$ is not bounded. The following invariants are
  valid on this net:
  \[
  \left\{\begin{array}{l}
    \Inv_1(m_0) = \{ m \mid m(p_1) =m_0(p_1) = 0 \} \\
    \Inv_2(m_0) = \{m \mid  \sum_{e \in E} m(e) + m(p_2) + \sum_{v \in
      V}(n-|E(v)|)\ m(v) =  \bfC_2^{m_0}\} \\
    \Inv_3(m_0) = \{m \mid \sum_{v \in V} m(v) + m(p_3) = \sum_{v \in V}
    m_0(v) = |V|\}
  \end{array}\right.
  \]
  In particular, $(\net,m_0)$ is unbounded if and only if the place
  $p_4$ can become unbounded. This is equivalent to enabling the
  transition $k\, p_3 \to p_4 + k\, p_3$, that is to have $k$ tokens
  or more in place $p_3$. This is only possible if we can find a
  subset $X \subseteq V$ of cardinality at least $k$ such that all for
  all $x,y \in X$, $E(x) \cap E(y) = \emptyset$. That is, $X$ needs to
  be an independent set of cardinality at least $k$.

  This concludes the coNP-hardness proof.

  Note that in the above proof, the unboundedness of the net is
  equivalent to the reachability of a marking where $p_4$ has at least
  one token. Hence the NP-hardness of the reachability and of the
  coverability problems follow.
\end{proof}

\section{Proofs of Section~\ref{sec:quantitative}}
\label{app:quantitative}

We provide here full proofs of the results mentioned in Section~\ref{sec:quantitative}
about quantitative results in stochastic open \pitrois-nets $(\calN,\lambda,m_0)$ with a live initial marking.
Recall that the set $P$ of places was partitioned between one set $X = \bigcup_{i \leqslant N-2} P_i \cup P_{N-1}^{\neg \max}$
and one set $Y = P_{N-1}^{\max} \cup P_N$.

\subsection{Ergodicity analysis}
\label{app:ergodic}

\begin{lemma}
\label{lem:bounded}
Let $\mathcal{F}$ be the family that contains the vectors
$p$ (for $p \in P_N$ and $\cin(p) = 0$), $\cin(p) q + |\cin(q)| p $ (for $p, q \in P_N$ and $\cin(p) > 0 > \cin(q)$), and
$q + |\cin(q)| p $ (for $p \in P_{N-1}^{\max}$, $q \in P_N$ and $0 > \cin(q)$), and let $\calL$ be the lattice generated by $\calF$.
In addition, let $m_0$ be a live marking, and let $\Inv(m_0)$ denote the set $\bigcap_{i=1}^{N-1} \Inv_i(m_0)$.

There exists a polynomially large integer $\bfG(m_0)$ such that
\[\{m_0\} + \calL \subseteq \calR(m_0) \subseteq \Inv(m_0) \subseteq \calG + \calL,\]
where $\calG$ is the set of markings $m$ such that $\|m\|_\infty \leqslant \bfG(m_0)$.
\end{lemma}

\begin{proof}
Theorem~\ref{theo:reach} states that $\calR(m_0) = \calI(m_0) \cap \Live$, where
$\Live = \bigcap_{i=1}^{N-1} \Live_i$ is the set of all live markings.
This already proves the inclusion relation $\calR(m_0) \subseteq \calI(m_0)$.

Second, for all $f \in \calF$, one checks easily the inclusions
$\{f\} + \calI(m_0) \subseteq \calI(m_0)$ and $\{f\} + \Live \subseteq \Live$.
It follows immediately that $\{m_0\} + \calL \subseteq \calR(m_0)$.

Finally, let $\bfG(m_0) = (\bfC_\infty^{m_0} + \bfE) (|P| + 1)$, where
$\bfE = \max\{\pot(p) \mid p \in P\} \leqslant \bfW |P|$ and
$\bfC_\infty^{m_0} = \max_{i = 1}^{N-2} \bfC_i^{m_0} \leqslant \bfE |P|$.
We prove that every marking $m \in \Inv(m_0)$ either belongs to the set
$\calG = \{m \mid \|m\|_\infty \leqslant \bfG(m_0)\}$
or is (componentwise) larger than some element of $\calF$.
Indeed, consider some place $p$.

If $p \in P_i$ for some $i \leqslant N-2$,
and since $\cin$ taked only non-negative entries on $P_{i+1}$, it follows that
$m(p) \leqslant m \cdot (P_i + \cin(P_{i+1})) = \bfC_i^{m_0} \leqslant \bfG(m_0)$.
Similarly, if $p \in P_{N-1}^{\neg\max}$, then
$m(p) \leqslant \cin(p) m(p) \leqslant m \cdot (P_{N-2} + \cin(P_{N-1})) = \bfC_{N-1}^{m_0} \leqslant \bfG(m_0)$.
This even proves that $m(p) \leqslant \bfC_\infty^{m_0}$ for all places $p \in Y$.

Assume now that $p \in X = P_N \cup P_{N-1}^{\max}$ and that $m$ is not larger than any element of $\calF$:
\begin{itemize}
 \item if $p \in P_N$ and $\cin(p) = 0$, we clearly have $m(p) = 0$;
 \item if $p \in P_N$, $\cin(p) < 0$, and $m(p) \geqslant \bfE$, then
 we must have $m(q) < \bfE$ for all places $q \in P_N$ such that $\cin(q) > 0$, and
 for all places $q \in P_{N-1}^{\max}$; it follows that
 \[\bfC_{N-1}^{m_0} \leqslant m \cdot P_{N-1} + \sum_{q \in P_N \text{ s.t. } \cin(q) > 0} m(q) \cin(q) - m(p) \leqslant
 (\bfC_\infty^{m_0} + \bfE) |P| - m(p),\] whence $m(p) \leqslant (\bfC_\infty^{m_0} + \bfE) |P| + \bfC_\infty^{m_0}$;
 \item if $p \in P_N$ and $\cin(p) > 0$, or if $p \in P_{N-1}^{\max}$, and if $m(p) \geqslant \bfE$, then
 we must have $m(q) < \bfE$ for all places $q \in P_N$ such that $\cin(q) < 0$; it follows that
 \[\bfC_{N-1}^{m_0} \geqslant m(p) + \sum_{q \in P_N \text{ s.t. } \cin(q) < 0} m(q) \cin(q) \geqslant
 m(p) - \bfE |P|,\] whence $m(p) \leqslant \bfE |P| + \bfC_\infty^{m_0}$.
\end{itemize}
This completes the proof.
\end{proof}

Observe that Lemma~\ref{lem:bounded} proves that all places $p \in X$ are bounded.
In addition, the following criterion for characterising ergodicity then comes quickly.

\theoergodic*

\begin{proof}
First, observe that these requirements are equivalent to saying that
$\hat{v}(f) < 1$ for all vectors $f \in \calF$.
Now, assume that $\hat{v}(f) \geqslant 1$ for some $f \in \calF$.
Lemma~\ref{lem:bounded} states that $\{m_0\} + \{k ~ \delta \mid k \in \bbN\} \subseteq \{m_0\} + \calL \subseteq \mathcal{R}(m_0)$,
and it follows that $\|\mathbf{v}\| \geq \hat{v}(m_0) \sum_{k \geq 0} \hat{v}(\delta)^k = +\infty$.

Conversely, assume that $\hat{v}(f) < 1$ for all $f \in \calF$.
Lemma~\ref{lem:bounded} then proves that $\calR(m_0) \subseteq \calG + \calL$, whence
\[\|\mathbf{v}\| \leqslant \sum_{g \in \calG, \ell \in \calL} \hat{v}(g + \ell) \leqslant
\hat{v}(\calG) \hat{v}(\calL) = \hat{v}(\calG) \prod_{f \in \calF} \frac{1}{1-\hat{v}(f)} < +\infty.\]
\end{proof}

\subsection{Computing the steady-state distribution}
\label{sec:app-compute}

We provide here full proofs for the results mentioned in Section~\ref{sec:computing}.

\productlem*

\begin{proof}
For every marking $m \in \calR(m_0)$, we denote by $\mathfrak{m}_X$ and $\mathfrak{m}_Y$
its respective classes in $\mathfrak{M}_X$ and $\mathfrak{M}_Y$,
and by $c(m)$ the integer $m \cdot P_{N-1}^{\neg\max}$.
It comes at once that
$\mathfrak{m}_X \in \calC_{m_0}(c(m))$ and that
$\mathfrak{m}_Y \in \calD_{m_0}(c(m))$.

Furthermore, we may prove that $0 \leqslant c(m) \leqslant |P| ~ \bfW ~ \|m_0\|$.
Indeed, by construction we have $0 \leqslant c(m)$.
Then, if $\bfW = 0$, then $P_{N-1}^{\neg \max} = \emptyset$, hence $c(m) = 0$.
Therefore, we assume that $\bfW \geqslant 1$.
Observe that $\pot(p) \leqslant |P| ~ \bfW$ for all places $p \in P$,
and therefore that $\cin(p) \leqslant |P| ~ \bfW$ as well.
Since $\cin(p) \geqslant 1$ for all $p \in P_{N-1}^{\neg \max}$, it follows that
\begin{align*}
c(m) =~ & m \cdot P_{N-1}^{\neg \max} \leqslant m \cdot \cin(P_{N-1}^{\neg \max}) \leqslant
m \cdot P_{N-2} + m \cdot \cin(P_{N-1}) = \bfC_{m_0}^{N-2} \\
\leqslant~ & m_0 \cdot P_{N-2} + m_0 \cdot \cin(P_{N-1}^{\neg \max}) \leqslant |P| ~ \bfW ~ \|m_0\|.
\end{align*}

Conversely, for all $c \geqslant 0$ and for all classes
$\mathfrak{m}_X \in \calC_{m_0}(c)$ and $\mathfrak{m}_Y \in \mathcal{D}_{m_0}(c)$,
the intersection $\mathfrak{m}_X \cap \mathfrak{m}_Y$ is a singleton set $\{m\}$.
Let $m_X$ and $m_Y$ be markings in $\calR(m_0)$ such that $m_X \in \mathfrak{m}_X$ and
$m_Y \in \mathfrak{m}_Y$: $m$ and $m_X$ coincide on $X$, and $m$ and $m_Y$ coincide on $Y$.
By construction, both $m$ and $m_X$ satisfy the invariants $\Inv_i(m_0)$ and $\Live_i$ for $i \leqslant N-2$.
Since $m_Y$ satisfies the invariant $\Inv_{N-1}(m_0)$, observe that
$m_Y \cdot P_{N-1}^{\neg \max} = c = m_X \cdot P_{N-1}^{\neg \max} = m \cdot P_{N-1}^{\neg \max}$.
Hence, both $m$ and $m_Y$ satisfy the invariants $\Inv_{N-1}(m_0)$ and $\Live_{N-1}$,
which proves that $r \in \calR(m_0)$.

Lemma~\ref{lemma:product} then follows from the inequality $0 \leqslant c(m) \leqslant |P| ~ \bfW ~ \|m_0\|$
and from the relation $\hat{v}(m) = \hat{v}(\mathfrak{m}_X) \hat{v}(\mathfrak{m}_Y)$.
\end{proof}

\procomputec*
\label{proof-computec}

\begin{proof}
We compute the sum $\sum_{\mathfrak{m}_X \in \calC_{m_0}(c)} \hat{v}(\mathfrak{m}_X)$,
by using a dynamic-programming approach.
This approach is similar to the procedure used in~\cite{HMN-fi13},
and it requires computing recursively a limited number of auxiliary sums
that generalise the sum that we want to compute.
Such auxiliary sums are of the form
$\sum_{\mathfrak{m}_1 \in Z} \hat{v}(\mathfrak{m}_1)$
for a class of sets $Z$ generalising the set $\calC_{m_0}(c)$.

We proceed as follows. First, we denote by $\kappa_i$ the cardinality of the set $P_i$,
for $i \leqslant N-2$, and by $\kappa_{N-1}$ the cardinality of $P_{N-1}^{\neg\max}$.
We also denote by $p_1^i,\ldots,p_{\kappa_i}^i$ the places in $P_i$ (or $P_i^{\neg\max}$ if $i = N-1$),
sorted by increasing value of potential, and we denote by $P(i,k)$
the set of places $\{p_1^i,\ldots,p_k^i\} \cup \bigcup_{j \leqslant i-1} P_j$.

Then, for all integers $1 \leqslant i \leqslant N-1$, $1 \leqslant j \leqslant k \leqslant \kappa_i$ and
$c, c' \geqslant 0$, we define $\overline{\calC}(i,j,k,c,c')$
as the set of those classes $\mathfrak{m}$ in $\mathfrak{M}_{P(i,k)}$
such that:
(i) $\mathfrak{m} \cdot \{p_1^i,\ldots,p_{j-1}^i\} = 0$,
(ii) if $j < \kappa_i$, then $\mathfrak{m}(p_j^i) \geqslant 1$,
(iii) $\mathfrak{m} \cdot \{p_1^i,\ldots,p_k^i\} = c$,
(iv) $\mathfrak{m} \cdot P_{i-1} + \mathfrak{m} \cdot \cin(\{p_1^i,\ldots,p_k^i\}) = c'$, and
(v) $\mathfrak{m}$ contains markings in $\mathcal{R}(m_0)$.

Observe that $P(N-1,\kappa_{N-1}) = X$ and that $\calC_{m_0}(c)$ is the disjoint union
$\bigsqcup_{j=1}^{\kappa_{N-1}} \overline{\calC}(N-1,j,\kappa,c,\bfC_{N-2}^{m_0})$.
We may further prove that $\overline{\calC}(i,j,k,c,c')$ is empty whenever $c$ or $c'$
does not belong to the set $\{0,\ldots,|P| ~ \bfW ~ \|m_0\|\}$.
Hence, our dynamic-programming approach consists in computing
sums over each of the (polynomially many) sets $\overline{\calC}(i,j,k,c,c')$.
It relies on the following recursive decomposition of $\overline{\calC}(i,j,k,c,c')$:
\begin{itemize}
 \item if $c' < \pot(p_j^i)$, if $c < 0$, or if $j = k < \kappa_i$ and $c = 0$, then $\overline{\calC}(i,j,k,c,c') = \emptyset$;
 \item if $j < k$, then we distinguish classes $\mathfrak{m}$ based on the value of $\mathfrak{m}(p_k^i)$,
 thereby obtaining the decomposition
 \[\overline{\calC}(i,j,k,c,c') = \bigsqcup_{u=0}^{c-1} \{\mathfrak{m} \in \mathfrak{M}_{P(i,k)} \mid
 \mathfrak{m}(p_k^i) = u \text{ and } \mathfrak{m}\!\restriction_{P(i,k-1)} \in \overline{\calC}(i,j,k-1,c-u,c'-u \cin(p_k^i))\};\]
 \item if $i \geqslant 2$, $j = k$ and either $c \geqslant 1$ or $j = \kappa_i$,
 then we have $\mathfrak{m}(p_k^i) = c$, and we distinguish classes $\mathfrak{m}$ based on on the least index $j'$, if it exists, such that
 $\mathfrak{m}(p_{j'}^{i-1}) \geqslant 1$, thereby obtaining the decomposition
 \[\overline{\calC}(i,j,k,c,c') = \bigsqcup_{j' = 1}^{\kappa_{i-1}} \left\{\mathfrak{m} \in \mathfrak{M}_{P(i,k)} \left|
 \begin{array}{l}\mathfrak{m}(p_k^i) = c \text{, }
 \{p_1^i,\ldots,p_{k-1}^i\} = 0 \text{, } \mathfrak{m} \cdot P_{i-1} \geqslant \pot(p_j^i), \\
 \mathfrak{m}\!\restriction_{P(i-1,\kappa_{i-1})} \in \overline{\calC}(i-1,j',\kappa_{i-1},c'-\cin(p_j^i) c,\bfC_{i-2}^{m_0})\end{array}\right\};\right.\]
 \item if $i = 1$, $j = k$ and either $c \geqslant 1$ or $j = \kappa_i$, then we obtain directly the equality
 $\overline{\calC}(i,j,k,c,c') = \{\mathfrak{m} \in \mathfrak{M}_{P(i,k)} \mid \mathfrak{m}(p_k^i) = c \text{ and }
 \mathfrak{m} \cdot \{p_1^i,\ldots,p_{k-1}^i\} = 0\}$.
\end{itemize}

Indeed, and denoting by $\calS(i,j,k,c,c')$ the sum
$\sum_{\mathfrak{m} \in \overline{\calC}(i,j,k,c,c')} \pi (\mathfrak{m})$,
it follows immediately that
\[\calS(i,j,k,c,c') = \begin{cases}
0 & \text{if $c < 0$ or $c' < \pot(p_j^i)$} \\ 
0 & \text{if ($c = 0$ and $j < \kappa_i$) or ($c' < \POT_i$ and $c > 0$)} \\
\bfS_1 + \mu_{p_k^i} \bfS_2 & \text{if $j < k$ and $c > 0$} \\
\mu_{p_k^i} \bfS_2 & \text{if $j = k$ and $c > \mathbf{1}_{j < \kappa_i}$} \\
\mu_{p_k^i} \bfS_3 & \text{if $j = k < \kappa_i$, $c = 1$ and $i \geqslant 2$} \\
\bfS_4 & \text{if $j = k = \kappa_i$, $c = 0$ and $i \geqslant 2$} \\
\mu_{p_k^i}^c & \text{if $j = k$, $i = 1$ and $c = \mathbf{1}_{j < \kappa_i}$}
\end{cases} \text{, where }\]
\begin{itemize}
\item $\bfS_1 = \calS(i,j,k-1,c,c')$;
\item $\bfS_2 = \calS(i,j,k,c-1,c'-\cin(p_k^i))$;
\item $\bfS_3 = \sum_{j'=1}^{\kappa_{i-1}} \calS(i-1,j',\kappa_{i-1},c' - \cin(p_k^i),\bfC_{i-2}^{m_0})$ and
\item $\bfS_4= \sum_{j'=1}^{\kappa_{i-1}} \calS(i-1,j',\kappa_{i-1},c',\bfC_{i-2}^{m_0})$.
\end{itemize}

These relations allow us to express the sum $\calS(i,j,k,c,c')$ in terms of polynomially
many sums of the form $\calS(i_2,j_2,k_2,c_2,c'_2)$, where the tuple $(i_2,k_2,c_2)$ is
smaller (for the lexicographic order) than the tuple $(i,k,c)$.
Hence, they provide us with a well-defined recursive procedure for computing the sums 
$\calS(i,j,k,c,c')$, which involves polynomially many arithmetic operations.
\end{proof}

We also extend the sketch of proof of Proposition~\ref{proposition:compute-d}
into a real proof as follows.

\procomputed*
\label{proof-computed}

\begin{proof}
We compute the sum $\sum_{\mathfrak{m}_Y \in \calD_{m_0}(c)} \hat{v}(\mathfrak{m}_Y)$
by using a dynamic-programming approach.
Let $a = |P_{N-1}^{\max}|$, $b = |P_N|$, and define $t$ and $u$ to be integers such that
(i) $\cin(p_i^N) > 0$ iff $1 \leqslant i \leqslant t$,
(ii) $\cin(p_i^N) = 0$ iff $t < i < u$ and
(iii) $\cin(p_i^N) < 0$ iff $u \leqslant i \leqslant b$.
In addition, for all $i \leqslant b$, let $\Delta_i = \max\{1,|\cin(p_i^N)|\}$,
and we consider below the lattice $\bbL = \bbZ^{a-1} \times \prod_{i=1}^b (\Delta_i \bbZ)$.
Finally, for the sake of readability, we will also use $\mu_{x,i}$ and $\mu_{y,i}$
as placeholders for $\mu_{p_{\kappa_{N-1}+i}^{N-1}}$ and $\mu_{p_i^N}$ respectively.

Then, for all classes $\mathfrak{m} \in \calD_{m_0}(c)$,
let $s(\mathfrak{m}) = \min\{i \leqslant b \mid \mathfrak{m}(i) \geqslant 1\}$,
or $s(\mathfrak{m}) = b$ if $\mathfrak{m} \cdot P_N = 0$.
We further partition the set $\calD_{m_0}(c)$ into finitely many sets
$\calD_{m_0}(c,s) = \{\mathfrak{m} \in \calD_{m_0}(c) \mid s(\mathfrak{m}) = s\}$,
for all integers $s \in \{1,\ldots,b\}$.

Now, we identify every class $\mathfrak{m} \in \calD(s,c)$ with the vector
$(x_1,\ldots,x_a,y_s,\ldots,y_b)$ defined by $x_i = \mathfrak{m}(p_{\kappa_{N-1}+i}^{N-1})$ and
$y_i = \Delta_i (\mathfrak{m}(p_i^N) - \mathbf{1}_{i = s} \mathbf{1}_{s \neq b})$.
Every such vector satisfies the identity
\[c + \sum_{i=1}^a x_i + \sum_{i=s}^t y_i - \sum_{i=u}^b y_i + \mathbf{1}_{s \neq b} \cin(p_s^N) = \bfC_{N-1}^{m_0},\]
hence we may project away the entry $x_a$ without loss of information about $\mathfrak{m}$.
We denote below by $\xy$ the vector $(x_1,\ldots,x_{a-1},y_s,\ldots,y_b)$.

Using Theorem~\ref{theo:reach}, we observe that a class $\mathfrak{m} \in \mathfrak{M}_X$ belongs
to the set $\calD(c,s)$ if and only if it satisfies the following four (in)equalities:
\begin{align*}
& \mathfrak{m} \cdot \{p_1^N,\ldots,p_{s-1}^N\} = 0, & &
\mathfrak{m}(p_s^N) \geqslant 1 \text{ if } s < \kappa_N, \\
& c + \mathfrak{m} \cdot (P_{N-1}^{\max} + \cin(P_N)) = \mathbf{C}_{N-1}^{m_0} \text{ and } & &
c + \mathfrak{m} \cdot P_{N-1}^{\max} \geqslant \pot_s \text{, where}
\end{align*}
where the integer $\pot_s$ is defined as $\min\{\pot(p_s^N),\pot(p_\ext)\}$.
Rewriting these inequalities in terms of the vector $\xy$ associated with the marking $\mathfrak{m}$,
it follows that the function $\mathfrak{m} \mapsto \xy$ is a bijection from the set $\calD(c,s)$ to the set
\[\left\{\bfxy \in \bbR^{a+b-1}_{\geqslant 0} \cap \bbL \left|
\begin{array}{l}\bfA + \sum_{j=u}^b y_j \geqslant \sum_{i=1}^{a-1} x_i + \sum_{j=s}^t y_j \\
B  + \sum_{j=u}^{\beta} y_j \geqslant \sum_{j=s}^t y_j \\
y_j = 0 \text{ for all } j < s\end{array}\right\},\right.\]
where $\bfA = \bfC_{N-1}^{m_0} - \mathbf{1}_{s \neq b} \cdot \cin(p_s^N) - c$,
$\bfB = \bfC_{N-1}^{m_0} - \mathbf{1}_{s \neq b} \cdot \cin(p_s^N) - \pot_s$,
and $\bfxy$ denotes the vector $(x_1,\ldots,x_{a-1},y_1,\ldots,y_b)$.

In addition, it is easy to see that $\hat{v}(\mathfrak{m})$ is equal to the quantity $\hatt{v}(\xy)$,
where $\hatt{v}$ is the product-form function defined by
\[\hatt{v} : \bfxy \mapsto \bfP \prod_{i=1}^{a-1} \nu_{x,i}^{x_i} \prod_{i=s}^{b} \nu_{y,i}^{y_i},\]
where $\bfP = \mu_{x,a}^{\bfA}$, $\nu_{x,i} = \mu_{x,i} / \mu_{x,a}$,
$\nu_{y,i} = \mu_{y,i}^{1/\Delta_i} / \mu_{x,a}$ for $s \leqslant i \leqslant t$,
$\nu_{y,i} = \mu_{y,i}$ for $t < i < u$ and
$\nu_{y,i} = \mu_{y,i}^{1/\Delta_i} \mu_{x,a}$ for $u \leqslant i \leqslant b$.

Like in the proof of Proposition~\ref{proposition:compute-c},
our dynamic-programming approach is based on generalising the set
$\{\xy \mid \mathfrak{m} \in \calD(c,s)\}$.
Aiming towards this direction,
let $\mathbf{1}_{x,1},\ldots,\mathbf{1}_{x,a-1}$, $\mathbf{1}_{y,1},\ldots,\mathbf{1}_{y,b}$
be the canonical basis of the vector space $\bbR^{a+b-1}$,
and let $\mathbf{0}$ denote the zero vector.
For all integers $A \in \{-|\bfA|,\ldots,|\bfA|\}$ and
$B \in \{-|\bfB|,\ldots,|\bfB|\}$, all integers $1 \leqslant \alpha \leqslant a-1$ and $1 \leqslant \gamma \leqslant \beta \leqslant b$,
and all vectors $\bfw \in \{\lambda ~ \mathbf{1}_{y,\gamma} \mid 0 \leqslant \lambda < \Delta_\gamma\} \cup
\{\lambda ~ \mathbf{1}_{y,\beta} \mid 0 \leqslant \lambda < \Delta_\beta\}$, we set
\[\overline{\calD}(A,B,\bfw,\alpha,\beta,\gamma) = \left\{\xy \in \bbR^{a+b-1}_{\geqslant 0} \left|
\begin{array}{l}A + \sum_{j=u}^{\beta} y_j \geqslant \sum_{i=1}^{\alpha} x_i + \sum_{j=\gamma}^t y_j \\
B  + \sum_{j=u}^{\beta} y_j \geqslant \sum_{j=\gamma}^t y_j \\
\xy - \bfw \in \bbL \\
x_i = 0 \text{ for all } i > \alpha \\
y_j = 0 \text{ for all } j < \gamma \text{ and all } j > \beta
\end{array}\right\}\right.\]

Observe that this family of sets generalises the set $\{\xy \mid \mathfrak{m} \in \calD(c,s)\}$,
which is equal to $\overline{\calD}(\bfA,\bfB,\mathbf{0},a-1,b,s)$.
Moreover, by construction, $\bfA$ and $\bfB$ are polynomially bounded.
Therefore, it remains to evaluate polynomially many sums of the form
$\calU(A,B,\bfw,\alpha,\beta,\gamma) = \sum_{\xy \in \overline{\calD}(A,B,\bfw,\alpha,\beta,\gamma)} \hatt{v}(\xy)$.
Observing that the value of the vector $\bfw$ is useful only modulo $\bbL$,
we identify below a vector $\bfw = \lambda ~ \mathbf{1}_{y,j}$ with the unique vector
$\bfw' \in \{\lambda' ~ \mathbf{1}_{y,j} \mid 0 \leqslant \lambda' \leqslant \Delta_j-1\}$
such that $\lambda \equiv \lambda'  \!\!\!\mod{\Delta_j}$.

Our recursive evaluation works by incrementally eliminating the variables
that appear in the expression of the vectors $\xy \in \overline{\calD}(A,B,\bfw,\alpha,\beta,\gamma)$
that we consider: we cancel the coordinates of $\xy$ one by one, in a clever order,
in order to obtain a polynomial-time evaluation of the associated sums $\calU$.
The variables $y_\beta,y_{\beta-1},\ldots,y_u$ must be eliminated in this order, and so must
the variables $y_\gamma,y_{\gamma+1},\ldots,y_t,x_\alpha,x_{\alpha-1},\ldots,x_1$.
Yet, if $v_1$ and $v_2$ are two variables chosen respectively from these two families,
it is possible to eliminate the variable $v_1$ before, after, or at the same time as the variable $v_2$.
In particular, we consider exponentially many interleavings for the elimination order of our variables:
due to our dynamic programming approach, considering all these interleaving will be made at little cost.

In practice, we may consider two stacks, each containing vectors of the canonical basis.
The stack $\bfS_1$ contains the vectors $\mathbf{1}_{y,u},\mathbf{1}_{y,u+1},\cdots,\mathbf{1}_{y,\beta}$ (from bottom to top) and
the stack $\bfS_2$ contains the vectors $\mathbf{1}_{x,1},\mathbf{1}_{x,2},\cdots,\mathbf{1}_{x,\alpha},
\mathbf{1}_{y,t},\mathbf{1}_{y,t-1},\cdots,\mathbf{1}_{y,\gamma}$.

 \begin{itemize}
  \item If both $\bfS_1$ and $\bfS_2$ are empty, then $\overline{\calD}(A,B,\bfw,\alpha,\beta,\gamma)$ is either empty
  (if $A < 0$ or $B < 0$ or $\bfw \neq \mathbf{0}$) or equal to the set
  $\{\bfxy \in \bbN^{a+b-1} \mid x_i = 0 \text{ for all } i \text{ and } y_j = 0 \text{ for all } j < \gamma \text{ or } j > \beta\}$.
  It follows directly that
  $\calU(A,B,\bfw,\alpha,\beta,\gamma) =
  \mathbf{1}_{A \geqslant 0 \text{, } B \geqslant 0 \text{ and } \bfw = \mathbf{0}} \bfP \prod_{j=\gamma}^\beta 1/(1-\nu_{y,j})$.
  
  \item If only $\bfS_1$ is non-empty, let $\mathbf{1}_{y,\beta}$ be its top element, and let us partition
  the set $\bbN$ of natural integers into its substets $\{0\}, \{1\}, \cdots, \{\max(|A|,|B|)\}$ and $\{z \in \bbN \mid z> \max(|A|,|B|)\}$.
  We distinguish the vectors $\xy \in \overline{\calD}(A,B,\bfw,\alpha,\beta,\gamma)$
  based on the subset of $\bbN$ to which their coordinate $y_\beta$ belongs.
  This allows us to reduce the computation of $\calU(A,B,\bfw,\alpha,\beta,\gamma)$ to the computation of finitely many sums
  of the form $\calU(A',B',\bfw',\alpha',\beta-1,\gamma')$.
  
  \item If only $\bfS_2$ is non-empty, let $\mathbf{1}_{y,\gamma}$ (or $\mathbf{1}_{x,\alpha}$) be its top element.
  We distinguish the vectors $\xy \in \overline{\calD}(A,B,\bfw,\alpha,\beta,\gamma)$
  based on the value of their coordinate $y_\gamma$ (or $x_\alpha$).
  This value may not exceed $\max(|A|,|B|)$, which reduces the computation of $\cal(A,B,\bfw,\alpha,\beta,\gamma)$
  to that of sums of the form $\overline{\calD}(A,B,\bfw,\alpha,\beta,\gamma+1)$ or $\overline{\calD}(A,B,\bfw,\alpha+1,\beta,\gamma)$.
  
  \item If both $\bfS_1$ and $\bfS_2$ are non-empty, let $\mathbf{1}_{y,\beta}$ be the top element of $\bfS_1$, and
  let $\mathbf{1}_{y,\gamma}$ (or $\mathbf{1}_{x,\alpha}$) be the top element of $\bfS_2$.
  Let us partition
  the set $\bbN$ of natural integers into its substets $\{0\}, \{1\}, \cdots, \{|B|\}$ and $\{z \in \bbN \mid z > |B|\}$.
  We distinguish the vectors $\xy \in \overline{\calD}(A,B,\bfw,\alpha,\beta,\gamma)$ based on
  which of the entries $y_\beta$ and $y_\gamma$ (or $x_\alpha$) is the smallest.
  If the top element of $\bfS_2$ is $\mathbf{1}_{x,\alpha}$, we further distinguish vectors $\xy$ based on
  which subset of $\bbN$ the integer $\min(y_\beta,x_\alpha)$ belongs to.
  This allows us to reduce the computation of $\calU(A,B,\bfw,\alpha,\beta,\gamma)$ to the computation of finitely many sums
  of the form $\calU(A',B',\bfw',\alpha-1,\beta,\gamma)$, $\calU(A',B',\bfw',\alpha,\beta-1,\gamma)$ or $\calU(A',B',\bfw',\alpha,\beta,\gamma+1)$.
  \end{itemize}

In page~\pageref{explicit}, this computation was explicitly carried in the case where $\bfS_1$ and $\bfS_2$ are non-empty,
with respective top elements $\mathbf{1}_{y,\beta}$ and $\mathbf{1}_{y,\gamma}$, and where $\bfw$ is of the form $\lambda ~ \mathbf{1}_{y,\beta}$.
Since other computations follow the exact same kind of reasoning, we omit detailing them,
and just present here the actual computation that they allow us to perform:
\begin{enumerate}
 \item If $\gamma \leqslant t$, $u \leqslant \beta$ and $\bfw = \lambda ~ \mathbf{1}_{y,\beta}$ for some $\lambda \in \{0,\ldots,\Delta_\beta-1\}$, then we have
 \begin{align*}
 & \calU(A,B,\bfw,\alpha,\beta,\gamma) = (\bfU_1 + \bfU_2 - \bfU_3) / (1-\mu_{y,\gamma}^{\Delta_\beta} \mu_{y,\beta}^{\Delta_\gamma}) \text{, where} \\
 & \bfU_1 = \sum_{k=0}^{\Delta_\beta-1} \mu_{y,\gamma}^k \mu_{y,\beta}^{k \Delta_\gamma/\Delta_\beta}
 \calU(A,B,(\lambda + k \Delta_\gamma) ~ \mathbf{1}_{y,\beta} ,\alpha,\beta,\gamma+1), \\
 & \bfU_2 = \sum_{k=0}^{\Delta_\gamma-1} \mu_{y,\gamma}^{(\lambda + k \Delta_\beta)/\Delta_\gamma} \mu_{y,\beta}^{k + \lambda/\Delta_\beta}
 \calU(A,B, (\lambda + k \Delta_\beta) ~ \mathbf{1}_{y,\gamma},\alpha,\beta-1,\gamma) \text{ and} \\
 & \bfU_3 = \sum_{k=0}^{\Delta_\beta-1} \mu_{y,\gamma}^k \mu_{y,\beta}^{k \Delta_\gamma/\Delta_\beta}
 \mathbf{1}_{k \Delta_\gamma \equiv \lambda \!\!\!\!\!\mod{\Delta_\beta}}
 \calU(A,B,\mathbf{0},\alpha,\beta-1,\gamma+1).
 \end{align*}
 
 \item If $\gamma \leqslant t$, $u \leqslant \beta$ and $\bfw = \lambda ~ \mathbf{1}_{y,\gamma}$ for some $\lambda \in \{0,\ldots,\Delta_\gamma-1\}$, then, symmetrically, we have
 \begin{align*}
 & \calU(A,B,\bfw,\alpha,\beta,\gamma) = (\bfU_1 + \bfU_2 - \bfU_3) / (1-\mu_{y,\gamma}^{\Delta_\beta} \mu_{y,\beta}^{\Delta_\gamma}) \text{, where}  \\
 & \bfU_1 = \sum_{k=0}^{\Delta_\gamma-1} \mu_{y,\gamma}^{k \Delta_\beta/\Delta_\gamma} \mu_{y,\beta}^k
 \calU(A,B,(\lambda - k \Delta_\beta) ~ \mathbf{1}_{y,\gamma} ,\alpha,\beta-1,\gamma), \\
 & \bfU_2 = \sum_{k=0}^{\Delta_\beta-1} \mu_{y,\gamma}^{k + \lambda/\Delta_\gamma} \mu_{y,\beta}^{(\lambda + k \Delta_\gamma)/\Delta_\beta}
 \calU(A,B,(\lambda + k \Delta_\gamma) ~ \mathbf{1}_{y,\beta},\alpha,\beta,\gamma+1) \text{ and} \\
 & \bfU_3 = \sum_{k=0}^{\Delta_\gamma-1} \mu_{y,\gamma}^{k \Delta_\beta/\Delta_\gamma} \mu_{y,\beta}^k
 \mathbf{1}_{k \Delta_\beta \equiv \lambda \!\!\!\!\!\mod{\Delta_\gamma}}
 \calU(A,B,\mathbf{0},\alpha,\beta-1,\gamma+1).
 \end{align*}

 \item If $\gamma > t$, $\alpha \geqslant 1$ and $u \leqslant \beta$, then $\bfw = \lambda ~ \mathbf{1}_{y,\beta}$ for some $\lambda \in \{0,\ldots,\Delta_\beta-1\}$, hence we have
 \begin{align*}
 & \calU(A,B,\bfw,\alpha,\beta,\gamma) = \mathbf{1}_{B < 0} (\mathbf{1}_{\lambda = 0} \bfU_4 + \bfU_5) +
 \mathbf{1}_{B \geqslant 0} (\bfU_6 + \bfU_7) / (1-\mu_{x,\alpha}^{\Delta_\beta} \mu_{y,\beta}) \text{, where} \\
 & \bfU_4 = \calU(A,B,\bfw,\alpha,\beta-1,\gamma) - \calU(A,B,\bfw,\alpha-1,\beta-1,\gamma) \\
 & \bfU_5 = \mu_{x,\alpha} \mu_{y,\beta}^{1/\Delta_\beta}
 \calU(A,B+1,(\lambda + 1) ~ \mathbf{1}_{y,\beta},\alpha,\beta,\gamma) \\
 & \bfU_6 = \mu_{x,\alpha}^\lambda \mu_{y,\beta}^{\lambda/\Delta_\beta}
 (\calU(A,0,\mathbf{0},\alpha,\beta-1,\gamma) - \calU(A,0,\mathbf{0},\alpha-1,\beta-1,\gamma)) \\
 & \bfU_7 = \sum_{k=0}^{\Delta_\beta-1} \mu_{x,\alpha}^k \mu_{y,\beta}^{k/\Delta_\beta}
 \calU(A,0,(\lambda + k) ~ \mathbf{1}_{y,\beta},\alpha-1,\beta,\gamma).
 \end{align*}
 
 \item If $\gamma > t$, $\alpha = 0$ and $u \leqslant \beta$, then $\bfw = \lambda ~ \mathbf{1}_{y,\beta}$ for some $\lambda \in \{0,\ldots,\Delta_\beta-1\}$, hence we have
 \begin{align*}
 & \calU(A,B,\bfw,\alpha,\beta,\gamma) = \mathbf{1}_{A < 0 \text{ or } B < 0} (\mathbf{1}_{\lambda = 0} \bfU_8 + \bfU_9) +
 \mathbf{1}_{A \geqslant 0 \text{ and } B \geqslant 0} \bfU_{10} / (1-\mu_{x,a}^{\Delta_\beta} \mu_{y,\beta})\text{, where} \\
 & \bfU_8 = \calU(A,B,\bfw,\alpha,\beta-1,\gamma) \\
 & \bfU_9 = \mu_{x,a} \mu_{y,\beta}^{1/\Delta_\beta}
 \calU(\min\{A+1,0\},\min\{B+1,0\},(\lambda + 1) ~ \mathbf{1}_{y,\beta},\alpha,\beta,\gamma) \\
 & \bfU_{10} = \mu_{x,a}^\lambda \mu_{y,\beta}^{\lambda/\Delta_\beta}
 \calU(0,0,\mathbf{0},\alpha,\beta-1,\gamma).
 \end{align*}
 
 \item If $\gamma > t$, $\alpha = 0$ and $u > \beta$, then we already mentioned that
  \[\calU(A,B,\bfw,\alpha,\beta,\gamma) =
  \mathbf{1}_{A \geqslant 0 \text{, } B \geqslant 0 \text{ and } \bfw = \mathbf{0}} \bfP \prod_{j=\gamma}^\beta 1/(1-\mu_{y,j}).\]
 
 \item If $\gamma \leqslant t$ and $u > \beta$ then $\bfw = \lambda ~ \mathbf{1}_{y,\gamma}$ for some $\lambda \in \{0,\ldots,\Delta_\gamma-1\}$, hence we have
 \begin{align*}
 & \calU(A,B,\bfw,\alpha,\beta,\gamma) = \mathbf{1}_{A \geqslant 0 \text{ and } B \geqslant 0} (\mathbf{1}_{\lambda = 0} \bfU_{11} + \bfU_{12}) \text{, where} \\
 & \bfU_{11} = \calU(A,B,\bfw,\alpha,\beta,\gamma+1) \\
 & \bfU_{12} = \mu_{x,a}^{-1} \mu_{y,\gamma}^{1/\Delta_\gamma}
 \calU(A-1,B-1,(\lambda + 1) ~ \mathbf{1}_{y,\gamma}),\alpha,\beta,\gamma).
 \end{align*}
 
 \item If $\gamma > t$, $\alpha \geqslant 1$ and $u > \beta$, then $\bfw = \mathbf{0}$, hence we have
 \begin{align*}
 & \calU(A,B,\bfw,\alpha,\beta,\gamma) = \mathbf{1}_{A \geqslant 0 \text{ and } B \geqslant 0} \bfU_{13} \text{, where} \\
 & \bfU_{13} = \calU(A,B,\bfw,\alpha-1,\beta,\gamma) + \mu_{x,a}^{-1} \mu_{x,\alpha} \calU(A-1,B,\mathbf{0},\alpha,\beta,\gamma).
 \end{align*}
\end{enumerate}

Due to Theorem~\ref{theo:ergodic}, we know that since our net is ergodic then every of our denominators was positive,
hence these sums exist and may indeed be computed by using polynomially many arithmetic operations.
\end{proof}

\subsection{Detailed complexity analysis}

We proved above that the normalisation constant $\|\bfv\|$ can be computed
by performing polynomially many arithmetic operations.
In particular, if these operations are carried in constant time (e.g. by using floating-point arithmetic),
then the normalisation constant itself is of course computable in polynomial time.

However, if the firing rates $\lambda_b$ of the Petri net bags are rational numbers, then so is every constant $\mu_p$,
and so will be the normalisation constant $\|\bfv\|$. Hence, we must also look at how much time is needed to compute
an exact value of $\|\bfv\|$ as a rational number.
We assume below that every rate $\lambda_b$ is a rational number whose numerators and denominators
are $\bfK$-bit integers, for some integer $\bfK$.
Then, we will only need to prove that the numerators and denominators of the auxiliary numbers that we compute
are at most exponential in $|P|$, $\|m_0\|$, $\bfW$ and $\bfK$.

First, computing the visit rates $\mathbf{vis}(b)$ for each bag is done by inverting a
$P \times P$ matrix with coefficients chosen from a finite subset of $\bbQ$.
Hence, this is feasible in polynomial time, and both the numerator and
the denominator of each rate $\mathbf{vis}(b)$ are at most exponential in $|P|$ and $\bfK$.
Hence, so are the constants $\mu_p$ for all places $p \in P$: below, we represent them
as fractions $N_p / D_p$.

\begin{lemma}
Each sum $\calS(i,j,k,c,c')$ that we compute
is a rational number whose numerator and denominator
are at most exponential in $|P|$, $\bfG(m_0)$, $\bfW$ and $\bfK$.
\end{lemma}

\begin{proof}
It follows easily from Lemma~\ref{lem:bounded} that each set $\mathfrak{M}_{P(i,k)}$ contains
at most $(\bfG(m_0)+1)^{|P(i,k)|}$ classes of markings. Moreover, for each class $\mathfrak{m} \in \mathfrak{M}_{P(i,k)}$
the term $\hat{v}(\mathfrak{m})$ is equal to $\prod_{p \in P(i,k)} \mu_p^{\mathfrak{m}(p)}$.
Since $\mathfrak{m}(p) \leqslant \bfG(m_0)$ for all $p \in P(i,k)$, it follows that
$\hat{v}(\mathfrak{m})$ is a rational number whose denominator divides $\prod_{p \in X} D_p^{\bfG(m_0)}$,
and which is itself bounded above by $\prod_{p \in X} N_p^{\bfG(m_0)}$.
Consequently, every sum $\calS(i,j,k,c,c')$ is a rational number whose denominator divides $\prod_{p \in X} D_p^{\bfG(m_0)}$,
and which is itself bounded above by $\prod_{p \in X} (\bfG(m_0)+1) N_p^{\bfG(m_0)}$,
hence its denominator and numerator are of polynomial size.
\end{proof}

\begin{lemma}
Each sum $\calU(A,B,\mathbf{w},\alpha,\beta,\gamma)$ that we compute
is a rational number whose numerator and denominator
are at most exponential in $|P|$, $\bfG(m_0)$, $\bfW$ and $\bfK$.
\end{lemma}

\begin{proof}
We observe easily, using an induction on $\alpha + \beta - \gamma + |A| + |B|$, that each number
$\calU(A,B,\lambda ~ \mathbf{1}_{y,j},\alpha,\beta,\gamma)$ belongs to the set
$\mu_{y,j}^{\lambda/\Delta_j}  \frac{\bbN}{\bfZ'}$, where
\begin{align*}
& \bfZ' = \left( \prod_{i=\gamma}^t \prod_{j=u}^\beta \bfZ_{y,y}(i,j) \right) \left( \prod_{i=1}^{\alpha} \prod_{j=u}^\beta \bfZ_{x,y}(i,j) \right)
\left(\prod_{i=t+1}^{u-1} \bfZ_y(i) \right) \left(\prod_{p \in Y} \bfZ(p) \right),\end{align*}\vspace{-5mm}
\begin{align*}
& \bfZ_{y,y}(i,j) = D_{y,i}^{\Delta_j} D_{y,j}^{\Delta_i} (D_{y,i}^{\Delta_j} D_{y,j}^{\Delta_i} - N_{y,i}^{\Delta_j} N_{y,j}^{\Delta_i}),
& & \bfZ_y(i) = D_{y,i} - N_{y,i} \text{ and } \\
& \bfZ_{x,y}(i,j) = D_{x,i}^{\Delta_j} D_{y,j}  (D_{x,i}^{\Delta_j} D_{y,j} - N_{x,i}^{\Delta_j} N_{y,j}),
& & \bfZ(p) = (N_p D_p)^{|A|+|B|}.
\end{align*}

In addition, every sum $\calU(A,B,\lambda ~ \mathbf{1}_{y,j},\alpha,\beta,\gamma)$ is bounded above by
the sum $\bfU_0 = \sum_{\bfxy \in \bfD_0} \hatt{v}(\xy)$, where
\begin{align*}
\bfD_0 =~ & \{\bfxy \in \bbN^{a+b-1} \mid |\bfA|+ \sum_{j=u}^b y_j \geqslant \sum_{i=1}^{a-1} x_i + \sum_{j=1}^t y_j\} \\
\subseteq~ & \sum_{\bfxy \in \calF_1} \{k ~ \bfxy \mid 0 \leqslant k \leqslant |\bfA|\} + \sum_{\bfxy \in \calF_2} \{k ~ \bfxy \mid 0 \leqslant k\} \text{, where} \\
\calF_1 =~ & \{\mathbf{1}_{x,i} \mid 1 \leqslant i \leqslant a-1\} \cup \{\mathbf{1}_{y,j} \mid 1 \leqslant j \leqslant t\} \text{ and} \\
\calF_2 =~ & \{\bfxy + \mathbf{1}_{y,j} \mid \bfxy \in \calF_1 \text{ and } u \leqslant j \leqslant b\} \cup \{\mathbf{1}_{y,j} \mid t < j \leqslant b\}.
\end{align*}
It follows immediately that
\[\bfU_0 \leqslant \left( \prod_{\bfxy \in \calF_1} (|\bfA|+1) (1 + \hatt{v}(\bfxy)^{|\bfA|}) \right) \left( \prod_{\bfxy \in \calF_2} (1 - \hatt{v}(\bfxy))^{-1} \right).\]

Moreover, observe that $\hatt{v}(\mathbf{1}_{x,i}) = \mu_{x,i} / \mu_{x,a}$ for all $i \leqslant a-1$
and that $\hatt{v}(\mathbf{1}_{y,j}) \leqslant (1 + \mu_{y,j}) (\mu_{x,a} + 1 / \mu_{x,a})$ for all $j \leqslant b$.
Similarly, for all $\bfxy \in \calF_2$, we have the following upper bounds:
\[(1 - \hatt{v}(\bfxy))^{-1} \leqslant \begin{cases}
\Delta_j / (1- \mu_{x,i}^{\Delta_j} \mu_{y,j}) & \text{if } \bfxy = \mathbf{1}_{x,i} + \mathbf{1}_{y,j} \\
\Delta_i \Delta_j / (1- \mu_{y,i}^{\Delta_j} \mu_{y,j}^{\Delta_i}) & \text{if } \bfxy = \mathbf{1}_{x,i} + \mathbf{1}_{y,j} \\
1 / (1 - \mu_{y,j}) & \text{if } \bfxy = \mathbf{1}_{y,j} \text{ with } t < j < u \\
\Delta_j / (1 - \mu_{x,a}^{\Delta_j} \mu_{y,j}) & \text{if } \bfxy = \mathbf{1}_{y,j} \text{ with } u \leqslant j \leqslant b
\end{cases}\]
Recalling that $|\bfA|$ and the integers $\Delta_j$ are both polynomially bounded in $|P|$, $\bfW$ and $\bfG(m_0)$,
the upper bound $\bfU_0$ itself is at most exponential in $|P|$, $\bfW$, $\bfG(m_0)$ and $\bfK$.
This concludes the proof.
\end{proof}

\end{document}